\documentclass{article}
\usepackage{fullpage}
\usepackage{amsmath,bbm}
\usepackage[utf8]{inputenc}
\usepackage{graphicx}
\usepackage{comment}

\usepackage{booktabs}
\usepackage{threeparttable}

\usepackage{url}
\usepackage{setspace}
\usepackage[table,usenames,dvipsnames]{xcolor}
\usepackage{verbatim}

\definecolor{brightlavender}{rgb}{0.75, 0.58, 0.89}
\definecolor{lavenderpurple}{rgb}{0.59, 0.48, 0.71}
\definecolor{lavender}{rgb}{0.71, 0.49, 0.86}
\definecolor{lavenderpink}{rgb}{0.98, 0.68, 0.82}

\usepackage{color, colortbl}

\definecolor{c1}{RGB}{92,89,191}
\definecolor{c2}{RGB}{191,92,98}
\definecolor{c3}{RGB}{98,191,92}

\usepackage{authblk}
\usepackage{amsmath, amsfonts, amsbsy, amssymb, amsthm}
\usepackage{psfrag}
\usepackage{bm}
\usepackage{fullpage}
\usepackage{comment}
\usepackage[normalem]{ulem}              
\usepackage{comment}     

\newtheorem{theorem}{Theorem}
\newtheorem{corollary}{Corollary}
\graphicspath{ {figures/} }

\def \ie {\textit{i.e., }}
\def \eg {\textit{e.g., }}

\def \R {\mathbb{R}}

\def \Ha {\mathcal{H}^\alpha}
\def \Hb {\mathcal{H}^\beta}
\def \Ga {\mathcal{G}^\alpha}
\def \Gb {\mathcal{G}^\beta}

\definecolor{mydkgreen}{cmyk}{0.8,0.0,1.0,0.50}

\def \bx {\pmb{x}}

\def \bH {{\bf H}}

\title{{Symmetries and Cluster Synchronization in Multilayer Networks}}

\begin{document}


\author[1,2]{Fabio Della Rossa} 
\author[3]{Louis Pecora}
\author[1]{Karen Blaha}
\author[1]{Afroza Shirin}
\author[1]{Isaac Klickstein}
\author[1,4]{Francesco Sorrentino}

\affil[1]{University of New Mexico, Albuquerque, NM, 87131, USA}
\affil[3]{U.S. Naval Research Laboratory, Washington, DC 20375, USA}
\affil[2]{Politecnico di Milano, Piazza Leonardo da Vinci 32, 20133 Milano, Italy}
\affil[4]{Corresponding author fsorrent@unm.edu}

\maketitle

\begin{abstract}

{Real-world systems in epidemiology, social sciences, power transportation, economics and engineering are often described as 
multilayer networks.
Here we first define and compute the symmetries of multilayer networks, and then study the emergence of cluster synchronization in these networks.
We distinguish between independent layer symmetries which occur in one layer and are independent of the other layers and dependent layer symmetries which involve nodes in different layers.
We study stability of the cluster synchronous solution by decoupling the  problem into a number of independent blocks and assessing stability of each block through a Master Stability Function.
We see that blocks associated with dependent layer symmetries have a different structure than the other blocks, which affects the stability of clusters associated with these symmetries.
Finally, we validate the theory in a fully analog experiment in which seven electronic oscillators of three kinds are connected with two kinds of coupling.} 
\end{abstract}

\section*{Introduction}

{Real-world complex systems often contain heterogeneous components;
these components may interact in multiple ways via complex connectivity patterns, leading to complex dynamics.}
For example, the power grid contains both a transmission and a communication network, and we must model both to understand phenomena such as {cascade failures \cite{buldyrev2010catastrophic,korkali2017reducing} and blackouts duration \cite{rezai2017key,Rosato_2008bq}. In this context, a recent paper \cite{rezai2017key} has shown that increased coupling between the power and the communication layers can be beneficial in reducing vulnerability of the system as a whole.}
Neurological systems can be modeled with varying levels of complexity depending on the particular behavior of interest. 
Some approaches have a single kind of neurons which interact via multiple interaction layers~\cite{Pereda_2014,Song_2015,adhikari_2011,sorrentino2012synchronization};
central pattern generator approaches~\cite{goulding2009circuits, lodi2018design, danner2016central} have multiple kinds of neurons (identified by kind, such as extensor, flexor, synapses generator, etc.) which interact via multiple interaction layers.
The mathematical formalism introduced for these types of problems is the {multilayer network}~\cite{Kivela_2014,boccaletti2014structure,taylor2016enhanced}.

We are interested in symmetries and synchronization of multilayer networks of coupled oscillators;
synchronization is a collective behavior in which the dynamics of the network nodes converge on the same time-evolution.
The first study of synchronization in multilayer networks was presented in \cite{sorrentino2012synchronization,irving2012synchronization} and more recently in \cite{del2016synchronization}. 
These papers studied complete synchronization (all nodes synchronize on the same time-evolution) and consider diffusive coupling (the coupling matrices are Laplacian). 
Populations may exhibit more complex forms of synchronization, such as clustered synchronization (CS), where clusters of nodes exhibit synchronized dynamics but different clusters evolve on distinct time evolutions;
many papers consider CS in networks {formed of nodes all of the same type and connections all of the same type}
\cite{belykh2001cluster,belykh2008cluster,nicosia2013remote,pecora2014cluster,sorrentino2016complete,sorrentino2016approximate,cho2017stable,siddique2018symmetry}.

The emergence of synchronized clusters in a network requires that the set of the network nodes is partitioned into subsets of nodes, called equitable clusters, such that all the nodes in the same cluster receive the same total input from each one of the clusters, leading to the same temporal evolution~\cite{golubitsky2016rigid, belykh2003persistent,schaub2016graph}. 
These partitions often arise from the network symmetries~\cite{BDMacArthur_2008jk,klickstein2018generating}. A recent paper studied the effects of symmetry on collective dynamics in communities of coupled oscillators \cite{skardal2019symmetry}.
When equitable partitions arise from symmetry, a rigorous, group theory-based framework exists to analyze the stability of cluster synchronization in simple networks~\cite{pecora2014cluster,sorrentino2016complete,siddique2018symmetry}. 
However, cluster synchronization in multilayer networks has not been studied other than in a recent paper which focused on experimental observations {on a special multilayer network composed of nodes all of the same type}~\cite{blaha2019cluster}. 

{
This paper significantly advances the field of network dynamics by presenting one unified theory that addresses the problem of cluster synchronization in arbitrary multilayer networks, where each layer is formed of homogeneous units, but different layers have different units. {Our main contributions are twofold:  we define and compute the symmetries of multilayer networks and 
 we study the emergence of cluster synchronization in these networks analytically and experimentally. This involves analyzing stability of the cluster synchronous solution, where the stability problem is decoupled into a number of lower dimensional blocks of equations 
and a validation of the theory in a fully analog experiment with three layers, each one formed of different types of oscillators. 
We analyze CS in arbitrary multilayer networks, for which nodes are coupled through both intra-layer connections (inside each layer) and inter-layer connections (between different layers.) 
 We see that CS patterns of synchronization in multilayer networks are determined by the
symmetries of each individual layer but also the particular pattern of interconnectivity between layers. We see that only nodes from the same layer may be permuted among each other. However, other symmetries may involve simultaneous swaps of nodes in different layers. 
This leads to a classification of symmetries into {Independent Layer Symmetries} (ILS) which occur in one layer and are independent of swaps in other layers and {Dependent Layer Symmetries} (DLS) which require  moving nodes in different layers.} In what follows, we first present a general set of dynamical equations for the time evolution of  multilayer networks, we then define and compute the group of symmetries of multilayer networks with different kinds of nodes (each corresponding to a different layer) and different kinds of connections, for which stability of the CS solution is analyzed, and finally we apply the theoretical framework to predict the emergence of CS in an experimental system.}



\section*{Results}
\subsection*{Formulation and Dynamical Equations}\label{sec: intro}

Multilayer networks have different types of nodes interacting through different types of connections  \cite{Kivela_2014,boccaletti2014structure}. 
Previously defined subclasses of multilayer networks 
include multiplex~\cite{verbrugge1979multiplexity,sola2013eigenvector,gomez2013diffusion} and multidimensional~\cite{berlingerio2013multidimensional,coscia2013you} networks.
In a {multiplex network}, the same set of agents exists in all layers; for example, in a social system, each node is a person, each layer represents opinion on a topic, and links capture how social interactions influence a person's opinion on each topic. 
In a {multidimensional network}, different kinds of links connect the same set of nodes, all of the same type. 
For a deeper discussion of particular multilayer networks and how to reconcile each case with the general definition of a multilayer network, see~\cite{Kivela_2014,boccaletti2014structure}. {Sometimes the term `multilayer network' has been used in the literature to indicate generic networks formed of different types of nodes and/or connections. For example, in \cite{blaha2019cluster} a `multilayer network' has connections of different types but nodes all of the same type. Here we consider the most general situation for which both nodes and connections can be of different types, with the case of Ref.\ \cite{blaha2019cluster} remaining a special case of multilayer network.}

{
Reference \cite{tang2017master} introduced a mathematical model for the time evolution of a multiplex network.
Following previous studies of synchronization in different instances of multilayer networks \cite{sorrentino2012synchronization,irving2012synchronization,tang2017master,del2016synchronization,blaha2019cluster}, next we provide a general set of equations that describe the dynamics of a multilayer network. First we  define the sets of nodes and of connections (or interactions) of a multilayer network. 
The nodes are arranged in sets $\{X^\alpha, \alpha=1,\ldots,M\}$, where each individual set corresponds to a given layer of the multilayer network. 
	The uncoupled dynamics of all the $N^\alpha$ nodes in layer $X^\alpha$ is the same:  $\dot {\bf x}_i^\alpha = {\bf F}^\alpha({\bf x}_i^\alpha)$, $i=1,\ldots,N^\alpha$, ${\bf x}^\alpha \in \mathbb{R}^{n^\alpha}$.  The multilayer network has a total number of nodes equal to $N=\sum_\alpha N^\alpha$. The set of connections correspond to either
	 {intra-layer interactions} that connect nodes in the same layer or {inter-layer interactions} that connect nodes in different layers. 
 The intra-layer interactions inside layer $\alpha$ are described by an adjacency matrix $A^{\alpha\alpha}$,  to which is associated a nonlinear coupling function ${\bf H}^{\alpha\alpha}:\mathbb{R}^{n^\alpha}\mapsto\mathbb{R}^{n^\alpha}$ and  a coupling strength $\sigma^{\alpha\alpha}$. 
 The inter-layer interactions from layer $\beta$ to layer $\alpha$ are described by an  $N^\alpha\times N^\beta$ adjacency matrix $A^{\alpha\beta}$, to which is associated a nonlinear coupling function ${\bf H}^{\alpha\beta}:\mathbb{R}^{n^\alpha}\mapsto\mathbb{R}^{n^\beta}$ and  a coupling strength $\sigma^{\alpha\beta}$.}
	We assume throughout that all the couplings are undirected: $A^{\alpha\beta,\lambda}={A^{\beta\alpha,\lambda}}^T$, $\alpha,\beta=1,...,M$. 

{
We show an example of a multilayer network in Fig.\ 1{\textbf a}. This network has two layers, labeled $\alpha$ and $\beta$, with intra-connectivity described by the matrices $A^{\alpha \alpha}$ and $A^{\beta \beta}$ and inter-connectivity described by the matrix $A^{\alpha \beta}=(A^{\beta \alpha})^T$, all shown in {\textbf b}. 
Fig.\ 1{\textbf c} shows an independent layer symmetry. This symmetry  involves only nodes in layer $\alpha$; we can swap $\alpha$ nodes $1$ with $2$ and $\alpha$ nodes $3$ with $4$ without affecting layer $\beta$. Fig.\ 1{\textbf d} shows a dependent layer symmetry. This symmetry requires swapping nodes in both layer $\alpha$ and $\beta$; when we swap $\alpha$ nodes $2$ with $3$ and $\alpha$ nodes $1$ with $4$, we must also swap $\beta$ nodes $1$ and $3$.
Note that as dependent layer symmetries involve swapping nodes of different types, they are a characteristic feature of multilayer networks with different types of nodes. }

\begin{figure}[htb!]
    \includegraphics[width=\textwidth]{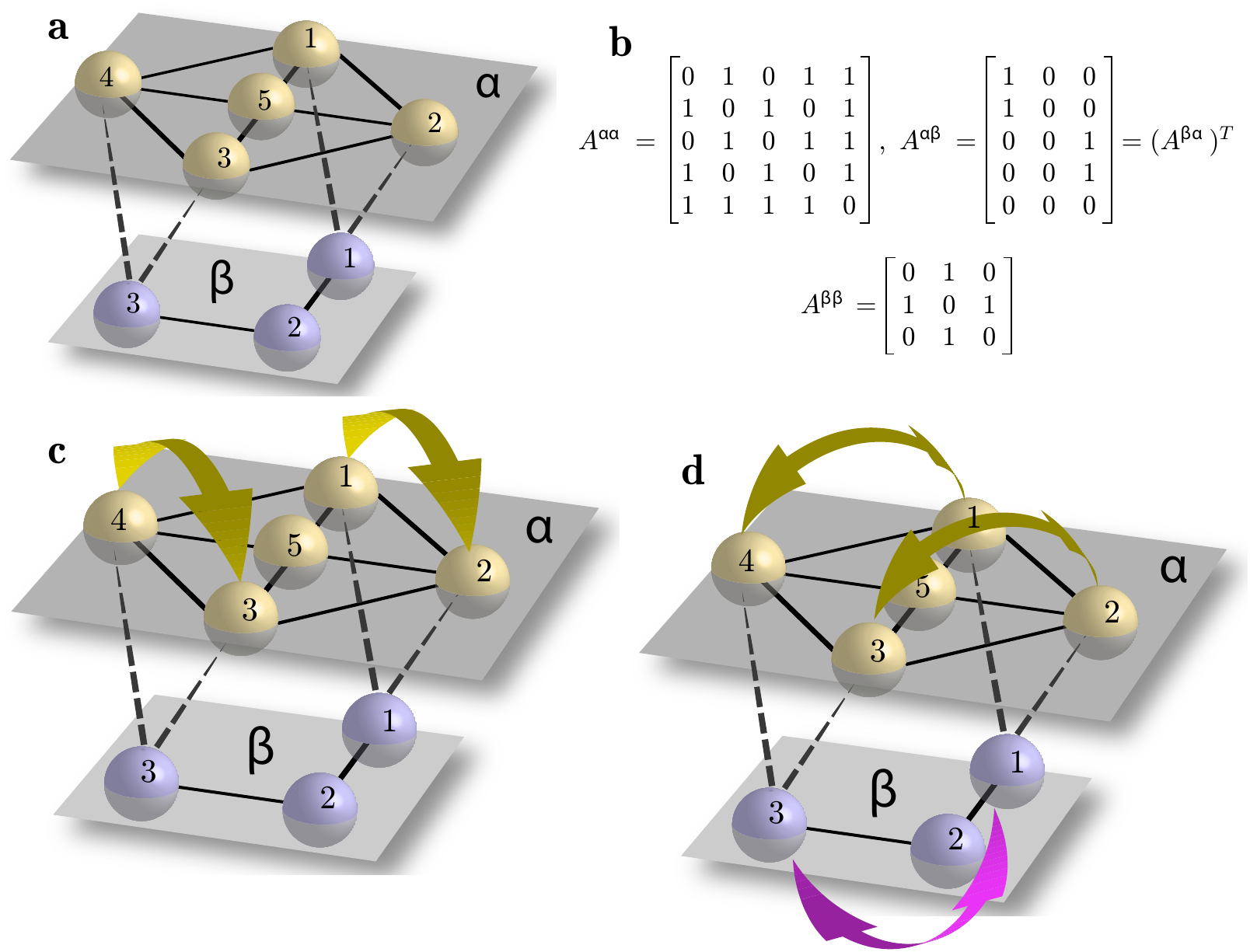}
    \caption{\textbf{Illustration of Independent Layer Symmetries and Dependent Layer Symmetries.} {\textbf{a}} A simple two layer network with intra-connectivity described by the matrices $A^{\alpha \alpha}$ and $A^{\beta \beta}$ and inter-connectivity described by the matrix $A^{\alpha \beta}=(A^{\beta \alpha})^T$ (in {\textbf{b}}). {\textbf{c}} Example of an Independent Layer Symmetry for the multilayer network in {\textbf{a}}. {\textbf{d}} Example of a Dependent Layer Symmetry for the multilayer network in {\textbf{a}}.}
    \label{fig:twolayer1}
\end{figure}


The dynamics of node $i$ in layer $\alpha$ of the multilayer network is governed by the following set of differential equations,

\begin{equation} \label{eq:gen_dyn}
    \dot {\bf x}_i^\alpha = {\bf F}^\alpha({\bf x}_i^\alpha) +
    \underbrace{ \sigma^{\alpha\alpha}
    \sum_{j=1}^{N^\alpha} A^{\alpha\alpha}_{ij} {\bf H}^{\alpha\alpha} ({\bf x}_j^\alpha)}_{\text{intra-layer couplings}}
    +
    \underbrace{\sum_{\beta\neq\alpha}
 \sigma^{\alpha\beta}
    \sum_{j=1}^{N^\beta} A^{\alpha\beta}_{ij} {\bf H}^{\alpha\beta} ({\bf x}_j^\beta)}_{\text{inter-layer couplings}},
\end{equation}
$i=1,...,N^{\alpha}$, $\alpha=1,...,M$.
The underlying assumption here is that functions that are labeled differently are different functions. 
By this we mean ${\bf F}^\alpha(x) \neq {\bf F}^\beta(x)$, for $\beta\neq \alpha$.

Next we introduce a vectorial notation. 
We combine the dynamical variables from each layer into a vector of vectors ${\bf x}^\alpha = [{\bf x}^\alpha_1,{\bf x}^\alpha_2,...,{\bf x}^\alpha_{N^\alpha}]$ (here and after, the comma stands for vertical stacking of vector and the square brackets stand for vector concatenation). 
Likewise, we can define, 
\begin{equation} 
{\bf F}^\alpha ({\bf x}^\alpha)=[{\bf F}^\alpha ({\bf x}^\alpha_1),{\bf F}^\alpha ({\bf x}^\alpha_2),...,{\bf F}^\alpha ({\bf x}^\alpha_{N^\alpha})], 
\end{equation}
\begin{equation} 
{\bf H}^{\alpha\alpha} ({\bf x}^\alpha)=[{\bf H}^{\alpha\alpha} ({\bf x}^\alpha_1),{\bf H}^{\alpha\alpha} ({\bf x}^\alpha_2),...,{\bf H}^{\alpha\alpha} ({\bf x}^\alpha_{N^\alpha})], 
\end{equation}
\begin{equation}
{\bf H}^{\alpha\beta} ({\bf x}^\beta)=[{\bf H}^{\alpha\beta} ({\bf x}^\beta_1),{\bf H}^{\alpha\beta} ({\bf x}^\beta_2),...,{\bf H}^{\alpha\beta} ({\bf x}^\beta_{N^\beta})]. 
\end{equation}
This notation suppresses the summations over the nodes in each layer. We obtain the set (one for each layer) of vector equations 
\begin{equation} \label{eq:CompactEqns_alpha}
    \dot {\bf x}^\alpha = {\bf F}^\alpha({\bf x}^\alpha) +
     \sigma^{\alpha\alpha}
     A^{\alpha\alpha} {\bf H}^{\alpha\alpha} ({\bf x}^\alpha)
    +
    \sum_{\beta\neq\alpha} \sigma^{\alpha\beta}
     A^{\alpha\beta} {\bf H}^{\alpha\beta} ({\bf x}^\beta), \quad \alpha=1,...,M.
\end{equation}

Two nodes $i$ and $j$ are synchronized if ${\bf x}_i(t)={\bf x}_j(t)\, \forall t$. 
Synchronous motions define invariant manifolds of Eqs. \eqref{eq:CompactEqns_alpha}, \ie ${\bf x}_i={\bf x}_j\Rightarrow \dot {\bf x}_i=\dot {\bf x}_j$. 
In order to  study cluster synchronization, we will look for symmetries in the multilayer network; these are the node permutations that leave system \eqref{eq:CompactEqns_alpha} unchanged.  
The {group of symmetries of the multilayer network} is the set of permutations of the nodes that do not change Eqs. \eqref{eq:CompactEqns_alpha} for $\alpha=1,...,M$. 
Below we explain how to find this group of symmetries in arbitrary multilayer networks, and we investigate the relations between intralayer and interlayer connections in order to preserve certain symmetries properties. 

\subsection*{Symmetries of Multilayer Networks}

Symmetries of complex networks have received recent attention from the scientific community. 
These symmetries influence network structural \cite{BDMacArthur_2008jk}, spectral \cite{macarthur2009spectral}, and dynamical properties \cite{sorrentino2019symmetries}, including cluster synchronization \cite{pecora2014cluster,sorrentino2016complete}. However, symmetries of multilayer networks have not been studied. 


Here we define the group of symmetries of general multilayer networks {with both nodes of different types and connections of different types}. 
Although we will eventually use computational methods to find the symmetry group of a particular network, it is best to understand what types of structures are imposed on the symmetry group of a multilayered network in general. The study of the symmetry group provides insight into the results of the numerical calculations 
and into whether what we find is general or just a property of the particular network we are analyzing. 
Next we develop some of the mathematical properties of the final permutation group structure of the multilayer systems.  As we will see, the interplay between the symmetries in each layer with the symmetries in other layers is subtle and leads to interesting structures in the final group of the whole network.

We first show that the multilayer structure imposes a block diagonal form on the permutations of the whole group. 
Since each layer of the network contains a different kind of node, symmetries are not allowed to move nodes from different layers. 
As a result, the symmetry group ${\mathcal G}$ of the multilayer network is represented by block diagonal permutation matrices, i. e., each $g \in {\mathcal G}$ is of the form,
\begin{equation} \label{eq:gblockdiag}
g=
\begin{pmatrix}
g_{\alpha}&0&0&\ldots  \\
0&g_{\beta}&0&\ldots  \\
0&0&g_{\gamma}&\ldots  \\
\vdots&\vdots&\vdots&\ddots 
\end{pmatrix},
\end{equation}
where $g_{\alpha}$ is a permutation matrix for layer $\alpha$, $g_{\beta}$ is a permutation matrix for layer $\beta$, and so on. 
In this section we show that the symmetries of the group ${\mathcal G}$ of the multilayer network are formed of the symmetries of the single-layers (i.e., $g_{\alpha} \in {\mathcal G}^{\alpha}$, $g_\beta \in {\mathcal G}^{\beta}$, ...) that are {compatible} with each other,  where compatibility is determined by the inter-layer couplings. This means that if we perform a permutation on layer $\alpha$, we must perform a permutation on each other layer (which in some cases may be the identity permutation)
by leaving the structure of the overall network unaltered.
As a result, compatibility links permutations from the groups of symmetries of each layer, so that we preserve the symmetry of the whole multilayer network. 
The major question we want to answer here is: how does compatibility relate permutations from different layers?

Let us consider a multilayer system with two layers $\alpha$ and $\beta$.  
Eq.\ \eqref{eq:CompactEqns_alpha} becomes,
\begin{equation}
    \begin{array}{rcl}
     \dot{\bf x}^\alpha &=& {\bf F}^\alpha({\bf x}^\alpha) + \sigma^{\alpha\alpha} A^{\alpha \alpha} {\bf H}^{\alpha\alpha} ({\bf x}^\alpha) + \sigma^{\alpha \beta} A^{\alpha \beta} {\bf H}^{\alpha \beta} ({\bf x}^\beta)
\\
     \dot{\bf x}^\beta  &=& {\bf F}^\beta({\bf x}^\beta) + 
     \sigma^{\beta\beta} A^{\beta \beta} {\bf H}^{\beta\beta} ({\bf x}^\beta) + \sigma^{\beta \alpha} A^{\beta \alpha} {\bf H}^{\beta \alpha} ({\bf x}^\alpha).
\end{array}
\end{equation}
In Methods, we show that we can determine the symmetry group of a multilayer network with $M>2$ layers, multiple edge types inside each layer and in between layers. 
Consider two permutations $g \in {\mathcal G}^\alpha$ and $h \in {\mathcal G}^\beta$. We know from Eq. \eqref{eq:gblockdiag} that a  symmetry of this multilayer network is in the form
\begin{equation} \label{eq:2Dblockdiag}
\begin{pmatrix}
g&0 \\
0&h
\end{pmatrix},
\end{equation}
so, which $g$'s and $h$'s are compatible?
For the answer, we consider the inter-layer couplings terms,
\begin{equation}
    \sigma^{\alpha \beta} A^{\alpha \beta} {\bf H}^{\alpha \beta} ({\bx}^\beta)
    \qquad \mbox{and} \qquad  
    \sigma^{\beta \alpha} A^{\beta \alpha} {\bf H}^{\beta \alpha} ({\bx}^\alpha).
\end{equation}
We say that symmetry-related nodes must be {flow invariant}.
That is, the symmetries must guarantee that synchronized nodes have equal dynamical variables when we include inter-layer coupling. 

Application of the two permutations to the equations of the multilayer network results in
\begin{equation}
    \begin{array}{rcl}
g{(\dot{\bf x}^\alpha)} &=& {\bf F}^\alpha(g{\bf x}^\alpha) + \sigma^{\alpha\alpha} A^{\alpha \alpha} {\bf H}^{\alpha\alpha} (g{\bf x}^\alpha)+ \sigma^{\alpha \beta} g A^{\alpha \beta} {\bf H}^{\alpha \beta} ({\bf x}^\beta)   \\
h(\dot{{\bf x}}^\beta) &=& {\bf F}^\beta(h{\bf x}^\beta) + 
     \sigma^{\beta\beta} A^{\beta \beta} {\bf H}^{\beta\beta} (h{\bf x}^\beta)+ \sigma^{\beta \alpha} h A^{\beta \alpha} {\bf H}^{\beta \alpha} ({\bf x}^\alpha),
\end{array}
\end{equation}
where the $g$ and $h$ permutations can be taken into the arguments of ${\bf F}^\alpha, {\bf H}^{\alpha\alpha}$ and ${\bf F}^\beta, {\bf H}^{\beta\beta}$, respectively since the functions operate sequentially on their vector arguments and we have used the properties that $g$ commutes with $A^{\alpha \alpha}$ and $h$ commutes with $A^{\beta \beta}$.

To achieve flow invariance, we need $g$ and $h$ to act on ${\mathbf x}^{\alpha}$ and ${\mathbf x}^{\beta}$, respectively, in the arguments of the interlayer coupling terms ${\bf H}^{\alpha \beta}$ and ${\bf H}^{\beta \alpha}$. 
Note that this does not require commutability ($g A^{\alpha\beta} = A^{\alpha\beta}g$) since we want to `exchange' $g$ for $h$ and vice-versa 
(as the interlayer coupling matrices are generally not square, we do not expect commmutability). 
More generally, conjugacy relations are the requirements for symmetry compatibility, \ie 
\begin{equation} \label{eq:conj}
    g A^{\alpha \beta}=A^{\alpha \beta} h \qquad \mbox{and} \qquad h A^{ \beta \alpha}=A^{\beta \alpha} g;
\end{equation}

\noindent which can be thought of as compatibility relations between permutations of the $\alpha$ and $\beta$ layers. 
This means that the $g$'s and the $h$'s must be paired properly to satisfy Eq.~\eqref{eq:conj}. 
In general, not all the $h$'s are compatible with all the $g$'s, as the conjugacy relations restrict compatible permutations to subgroups of ${\mathcal G}^\alpha$ and ${\mathcal G}^\beta$. The final group  ${\mathcal G}$ of the multilayer network is determined by the structure of these particular subgroups.

 The permutations that fulfill  the conjugacy relations \eqref{eq:conj} are defined by the following sets,

\begin{equation}\label{eq:Ha}
    \Ha= \{g \in \Ga | gA^{\alpha \beta}=A^{\alpha \beta}h\ \mbox{ and }  
    h A^{\beta \alpha}=A^{\beta \alpha}g \mbox{ for  some } \  h \in \Gb \}
\end{equation}

and

\begin{equation}\label{eq:Hb}
    \Hb= \{h \in \Gb | h A^{\beta \alpha}=A^{\beta \alpha} g\ \mbox{ and }  
    g A^{\alpha\beta}=A^{\alpha\beta}h \mbox{ for  some } \  g \in \Ga \}.
\end{equation}

An element $g \in \Ha$ is represented by an $N^\alpha \times N^\alpha$ matrix and an element $h \in \Hb$ is represented by an $N^\beta \times N^\beta$ matrix. It is important to note that in general the matrix $A^{\alpha\beta}\in\R^{N^\alpha \times N^\beta}$  will have nontrivial left and right null spaces. As a result, we may find more than one $g$ that satisfies $g A^{\alpha \beta}=A^{\alpha \beta} h$ for a given $h$ and vice versa.

In the Methods (Formal Proofs) we first prove that $\Ha$ is a subgroup of $\Ga$ and $\Hb$ is a subgroup of $\Gb$. Then we show that
for the case of undirected networks we only need either one of the following relationships to prove that $\Ha$ and $\Hb$ are subgroups,
\begin{equation} \label{eq:or}
    g A^{\alpha \beta}=A^{\alpha \beta} h \qquad \mbox{or} \qquad h A^{ \beta \alpha}=A^{\beta \alpha} g.
\end{equation}

We can now find the group of symmetries of the multilayer network $\mathcal G$ from the subgroups $\Ha$ and $\Hb$. 
First, recall that the permutations of $\mathcal G$ have block structure of Eq.\ \eqref{eq:2Dblockdiag}. 
While we use all permutations $g \in \Ha$ and $h\in\Hb$ to construct $\mathcal{G}$, we can only pair the permutations that satisfy the conjugacy relation in Eq. \eqref{eq:conj} (or the simpler version in Eq.\ \eqref{eq:or}). 

We define an equivalence relation $\sim$ between the elements of $\Ha$: $g \sim g'$ if $g A^{\alpha\beta}=g' A^{\alpha\beta}$. 
Analogously, $h \sim h'$ if $h A^{\beta\alpha}=h' A^{\beta\alpha}$. 
One can verify that $\sim$ is an equivalence relation as it is reflexive, symmetric, and transitive. 
We also see that if $g \sim g'$ and $g$ and $h$ are conjugate, then so are $g'$ and $h$, indicating that $\sim$ defines a partition of each subgroup $\Ha$ and $\Hb$ into disjoint subsets ({called equivalence classes}) $\mathcal K^{\alpha}_i$ and $\mathcal K^{\beta}_i$, $i=1,...,K$, respectively. 
Each subset $\mathcal K^{\alpha}_i$ contains all the permutations $g$ such that  $g A^{\alpha\beta}$ is equal to a given matrix $M_i$; 
correspondingly, each subset $\mathcal K^{\beta}_i$ contains all the permutations $h$ such that $h A^{\beta\alpha}$ is equal to {the matrix $M_i^T$.}
We can then construct the group of  symmetries of the multilayer network $\mathcal G$ as follows,
\begin{equation} \label{eq:2DGroup}
\mathcal G=\left\{
\begin{pmatrix}
g&0 \\
0&h
\end{pmatrix}
| g \in \mathcal K^{\alpha}_i \mbox{ and } h \in \mathcal K^{\beta}_i, \mbox{ for } i=1,...,K
\right\}
\end{equation}
Although the sets $\mathcal K^{\alpha}_i$ and $\mathcal K^{\beta}_i$ have a one-to-one correspondence, the elements of each do not and often differ in number, as shown {in Methods for the example multilayer network in Fig.\ 1.} 

Next, we present  properties of the equivalence classes $\mathcal K^\alpha_i$, which then leads to the definition of Independent Layer Symmetries (ILS).
Let $\mathcal K^\alpha_1$ be the equivalence class containing the identity. We can prove that $\mathcal K^\alpha_1$ is a normal subgroup of $\mathcal H^\alpha$.
Moreover, we can prove that all the $\mathcal K^\alpha_i$, $i \neq 1$, are left and right cosets of $\mathcal K^\alpha_1$. Formal proofs for each statement are included in Methods (Formal Proofs.)

The structure of the equivalence classes $\mathcal K^\alpha_i$ gives insight into the symmetries; this facilitates the calculations of the classes themselves and of the $T$ matrices for each layer (see Methods for the structure of the $T$ matrix). 
In particular the subgroups $\mathcal K^\alpha_1$, $\mathcal K^\beta_1$, ... identify a special set of symmetries we will refer to as 
 Independent Layer Symmetries. 
The following general relationship between the symmetry operations from $\mathcal K^\alpha_1$, $\mathcal K^\beta_1$,... and the layer structures clarifies why we refer to the elements of $\mathcal K^\alpha_1$, $\mathcal K^\beta_1$,... as Independent Layer Symmetries. Given any $g \in \mathcal K^\alpha_1$, then for each $h \in \mathcal K^\beta_1$,

\begin{table}[t!]
\centering
\caption{Symmetries of Real Multilayer Networks. For each network dataset, we include information on the  number of layers $M$, the number of nodes in each layer $ N^{\alpha}, \alpha=1,...M$ (the number of nodes $\mathcal{N}$ in all layers for multiplex networks), the number of edges $E$, the order of the automorphism group $|\mathcal{G}|$, and the cardinality of the largest orbital cluster $\max(|\mathcal{C}_k |)$. }
\begin{threeparttable}
	{\setlength\doublerulesep{ 2pt}
	\begin{tabular}{lclccc}
		\toprule\midrule
		Name                                                                                   & $M$ &      Number of nodes            &  $E$  &      $|\mathcal{G}|$       & $\max(| \mathcal{C}_k |)$ \\ \midrule\midrule
		\scriptsize{London Transport}\textsuperscript{**} \cite{de2014navigability}                      &  3  &      $\mathcal{N}=369$       &  441  &             4              &             3             \\
		\scriptsize{EU-AIR Transport} \cite{cardillo2013emergence}                        & 37  &      $\mathcal{N}=450$       & 3588  &  $1475532 \times 10^{31}$  &             4             \\
		\scriptsize{ARXIV NETSCIENCE}\textsuperscript{**} \cite{de2015identifying}                       & 13  &     $\mathcal{N}=14489$      & 59026 & $8349492 \times 10^{2173}$ &            16             \\
		\scriptsize{PIERRE AUGER Coauthorship}\textsuperscript{**} \cite{de2015identifying}              & 16  &      $\mathcal{N}=514$       & 7153  &   $19726 \times 10^{70}$   &            11             \\
		\scriptsize{CKM PHYSICIANS Social}$ ^\ast$ \cite{coleman1957diffusion}            &  3  &      $\mathcal{N}=246$       & 1551  &            120             &             2             \\
		\scriptsize{BOS Genetic}$ ^\ast$ \cite{de2015structural,stark2006biogrid}         &  4  &      $\mathcal{N}=321$       &  325  &     $2 \times 10^{90}$     &            25             \\
		\scriptsize{CANDIDA Genetic}$ ^\ast$ \cite{de2015structural,stark2006biogrid}     &  7  &      $\mathcal{N}=367$       &  397  &    $2 \times 10^{470}$     &            57             \\
		\scriptsize{DANIORERIO Genetic}$ ^\ast$ \cite{de2015structural,stark2006biogrid}  &  5  &      $\mathcal{N}=155$       &  188  &    $798 \times 10^{20}$    &            12             \\
		\scriptsize{{US power-grid}} \cite{nr}                             &  2  & $N^\alpha=4492; N^\beta=449$ & 6994  &    $5 \times 10^{152}$     &             9             \\ \midrule\bottomrule
	\end{tabular}
}
\label{tab:realnet}
\begin{tablenotes}
	\item[*] \footnotesize{Network treated as undirected.} \item[**] \footnotesize{Weighted Network but treated as unweighted.}
\end{tablenotes}
\end{threeparttable}
\end{table}

\begin{equation} \label{eq:ghindeplayersym}
 g A^{\alpha\beta}=A^{\alpha\beta} h = A^{\alpha\beta},
\end{equation}
since the equivalence class subgroups are both associated with the identity operation. This means that the symmetry operations from the equivalence class subgroups do not affect the interlayer coupling and, hence, those operations (permutations) in one layer do not affect the types of allowed dynamics in the other layer(s). In particular, the loss of a symmetry associated with the ILS subgroup in the dynamics, i.e. a symmetry-breaking bifurcation, will not alter the possible clusters in other layers, provided no symmetries in the other cosets ($\mathcal K^\alpha_i, i \neq 1$) are also broken. We will show simple examples of this in the Methods Section and when presenting the experiment.
Note that the stability of the clusters is a different issue, which will receive separate consideration.

To conclude, the ILS clusters are determined by the interplay between the intralayer couplings in each layer and the interlayer couplings between layers. In what follows we will refer to symmetries that are not ILS, i.e. they involve simultaneously swapping nodes in different layers as Dependent Layer Symmetries (DLS). 

From knowledge of the group of symmetries of the multilayer network, the nodes in each layer $\alpha$ can be partitioned into $L^\alpha$ {orbital} clusters, $\mathcal{C}^\alpha_1, \mathcal{C}^\alpha_2,..., \mathcal{C}^\alpha_{L^\alpha}$. Inside each layer, symmetries map into each other only nodes in the same orbital cluster.


{
In Table 1 we apply the techniques described in this section to compute the symmetries of several multilayer networks from datasets in the literature. For each network dataset, we include information on  the  number of layers $M$, the number of nodes in each layer $N^{\alpha}, \alpha=1,...M$ (the number of nodes $\mathcal{N}$ in all layers  for multiplex networks), the number of edges $E$, the order of the automorphism group $|\mathcal{G}|$, and the cardinality of the largest orbital cluster $\max(|\mathcal{C}_k |)$.
The main underlying assumption is that nodes are homogeneous inside each layer but nodes in different layers are not, so that nodes inside each layer can be symmetric, but nodes from different layers cannot. While we understand this is an oversimplification, available datasets do not typically include information on the attributes of the individual nodes,  and as a result, our symmetry analysis is solely based on the multilayer network structure [37]  and not on the specific node attributes. An extensive description of each one of the datasets is included in the Supplementary Note 1 and in the Supplementary Table 1. From Table 1 we see that several  real multilayer networks  possess very large numbers of symmetries.}

{In Supplementary Note 2 we study the emergence of symmetries in artificially generated multilayer networks, where each layer is a Scale Free (SF)  network and nodes from different layers are randomly matched to one another to obtain a multiplex network. We see that these networks typically do not display symmetries, except for the case that the power-law degree distribution exponents of the networks in each layer are low across the layers. It is interesting that {some of the} real multilayer networks {we have analyzed} display many more symmetries than {these} model multilayer networks. This observation motivated us to design a generating algorithm to construct multilayer networks with prescribed number of symmetries, which is presented in Supplementary Note 3.}

\subsection*{Stability Analysis}

In what follows we study stability of the cluster-synchronous solution for a general multilayer network described by Eqs.\ \eqref{eq:gen_dyn}.
 Given an orbital partition, we define the ${L^\alpha}\times {L^\alpha}$ intralayer quotient matrix $Q^{\alpha \alpha}$ such that for each pair of $\alpha$-clusters ($\mathcal{C}^\alpha_u,\mathcal{C}^\alpha_v$),
	\begin{equation}
	Q^{\alpha \alpha}_{uv}=\sum_{j\in\mathcal{C}^\alpha_v}A^{\alpha \alpha}_{ij}, \qquad i\in\mathcal{C}^\alpha_u,\quad u, v=1,2,\dots {L^\alpha}.
	\label{eq: quotient matrix}
	\end{equation}
	Analogously, we define the ${L^\alpha}\times {L^\beta}$ interlayer quotient matrix $Q^{\alpha \beta}$ such that for each pair of clusters, the first cluster ${\mathcal C}^\alpha_u$ from layer $\alpha$ and the second cluster ${\mathcal C}^\beta_v$ from layer $\beta$,
	\begin{equation}
	Q^{\alpha \beta}_{uv}=\sum_{j\in\mathcal{C}^\beta_v}A^{\alpha \beta}_{ij}, \qquad i\in\mathcal{C}^\alpha_u,\quad u=1,2,\dots {L^\alpha}, v=1,2,\dots {L^\beta}.
	\label{eq: quotient matrix}
	\end{equation}

We can thus write the equations for the time evolution of the  {quotient multilayer} network,

\begin{equation} \label{quot}
    \dot {\bf q}_u^\alpha = {\bf F}^\alpha({\bf q}_u^\alpha) +
    { \sigma^{\alpha\alpha}
    \sum_{v=1}^{L^\alpha} Q^{\alpha\alpha}_{uv} {\bf H}^{\alpha\alpha} ({\bf q}_v^\alpha)}
    +
    {\sum_{\beta\neq\alpha}
     \sigma^{\alpha\beta}
    \sum_{v=1}^{L^\beta} Q^{\alpha\beta}_{uv} {\bf H}^{\alpha\beta} ({\bf q}_v^\beta)},
\end{equation}
$\alpha=1,...,M$, $u=1,...,L^\alpha$. Note that  Eqs. \eqref{quot} provides  a mapping for each layer $\alpha$ from the node coordinates  $\{{\bf x}_i^\alpha \}$, $i=1,...,N^\alpha$ to the quotient coordinates $\{{\bf q}_u^\alpha \}$, $u=1,...,L^\alpha$, where ${\bf x}_i^\alpha(t) \equiv {\bf q}_u^\alpha(t)$ if $i \in \mathcal{C}_u^\alpha$.

We now linearize \eqref{eq:gen_dyn} about \eqref{quot},


\begin{equation} \label{interm}
     \delta \dot{{\bf x}}_i^\alpha = DF^\alpha({\bf q}_u^\alpha) \delta {{\bf x}}_i^\alpha +
    {\sum_{\beta}
     \sigma^{\alpha\beta}
    \sum_{j=1}^{N^\beta} A^{\alpha\beta}_{ij} DH^{\alpha\beta} ({\bf q}_v^\beta) \delta {\bf x}_j^\beta }, \qquad i=1,...,N^{\alpha},
\end{equation}
where again ${\bf q}_v^\beta$ is the time evolution of the quotient network node $v$ that node $j \in {\mathcal C}_v^\beta$ maps to.
Also note that in the above equation the summation in $\beta$ runs over both intra-layer connections and inter-layer connections.

For each layer $\alpha$, the above set of equations, can be written in vectorial form, by stacking together all the individual perturbations applied to the vectors
inside each layer, e.g., $\delta {\bf x}^{\alpha}=[\delta {\bf x}_1^{\alpha}, \delta {\bf x}_2^{\alpha},..., \delta {\bf x}_{N^{\alpha}}^{\alpha} ]$,

\begin{equation} \label{eq:stacked}
\delta \dot{\bf x}^\alpha =  \left(\sum_{u=1}^{L^\alpha} E_{u}^{\alpha} \otimes DF^\alpha({\bf q}_u^\alpha)
    \right)\delta {\bf x}^\alpha
     + \sum_\beta\left(\sigma^{\alpha\beta} (A^{\alpha\beta}\otimes I_{n^\alpha}) \sum_{u=1}^{L^\beta}    \left(E^\beta_{u} \otimes DH^{\alpha \beta}({\bf q}_u^{\beta})\right)\right) \delta {\bf x}^\beta, 
\end{equation}  
$\alpha=1,...,M,$ where each indicator matrix $E_{u}^\alpha$ has dimension $N^\alpha$ and is such that its diagonal entries are equal to one if node $i$ of layer $\alpha$ is in cluster ${\mathcal{C}_u^\alpha}$ and are equal to zero otherwise.

We then stack all the layers one above the other to form the vector $\delta {\bf x}=[\delta {\bf x}^{\alpha}, \delta {\bf x}^{\beta}, \dots ]$,
and we rewrite \eqref{eq:stacked} as 
\begin{equation} \label{eq:stacked2}
\begin{split}
     \delta \dot{\bf x} = & \left(
      \begin{bmatrix}
\sum_{u} E_{u}^{\alpha} \otimes DF^\alpha({\bf q}_u^\alpha)  & 0 &  \cdots\\
0 & \sum_{u} E_{u}^{\beta} \otimes DF^\beta({\bf q}_u^\beta) &  \cdots\\
\vdots & \vdots &  \ddots
\end{bmatrix} \right.  \\
\
 + & \left. \begin{bmatrix}
  \sum_u A^{\alpha\alpha} E^\alpha_{u}  \otimes D\hat{H}^{\alpha \alpha}({\bf q}_u^{\alpha})  & \sum_u A^{\alpha\beta}   E^\beta_{u}  \otimes D\hat{H}^{\alpha \beta}({\bf q}_u^{\beta}) &  \cdots\\
\sum_u A^{\beta\alpha} E^\alpha_{u}  \otimes D\hat{H}^{\beta \alpha}({\bf q}_u^{\beta}) &  \sum_u A^{\beta\beta} E^\beta_{u}  \otimes D\hat{H}^{\beta \beta}({\bf q}_u^{\beta}) &  \cdots\\
\vdots & \vdots &  \ddots
\end{bmatrix} \right)
 \delta {\bf x}, 
\end{split}
\end{equation}
    where $D\hat{H}^{\alpha \beta}= \sigma^{\alpha \beta} DH^{\alpha \beta} $.

 We are looking for a transformation that applied to Eq.\ \eqref{eq:stacked2}, leaves the first terms on the right hand side of \eqref{eq:stacked2} {unchanged} and decouples the second terms in independent blocks, independent of the $DF$ and the $DH$ terms as they vary in time. 




From knowledge of the group of symmetries of the multilayer network $\mathcal G$, we can compute
the irreducible representations (IRRs) of $\mathcal G$. For each layer $\alpha$ we can define an orthonormal transformation $T^\alpha$  to the  IRR coordinate system (see \cite{pecora2014cluster}). We can then construct the following block diagonal orthonormal matrix,
\begin{equation}
T=\bigoplus_\alpha T^\alpha,
\end{equation}
that maps the entire multilayer network to the IRR coordinate system.

A formal proof of the above particular structure of the matrix $T$ can be found in Methods (Properties of the Equivalence Classes and structure of the matrix $T$).
Intuitively, the matrix $T$ has a block diagonal structure, where each block corresponds to a layer, because only nodes from the same layer can be swapped by a symmetry.

Consider now the $N$-dimensional {supra-adjacency matrix},
 \begin{equation}
    A= \begin{bmatrix}
A^{\alpha \alpha}  & A^{\alpha \beta} &  \cdots\\
A^{\beta \alpha} & A^{\beta \beta} &  \cdots\\
\vdots & \vdots &  \ddots
\end{bmatrix}.
\end{equation}
The transformed $N \times N$ block diagonal matrix $B= T A T^{-1}$ is a direct sum $\oplus_{r=1}^R I_{d_r} \otimes \hat{B}_r$, where $\hat{B}_r$ is a (generally complex) $p_r \times p_r$ matrix
with $p_r$ the multiplicity of the $r$th IRR in the permutation
representation, $R$ the number of IRRs present and $d_r$
the dimension of the $r$th IRR, so that
$\sum_r d_r p_r= N$.
The matrix $T$ contains information on which perturbations affecting different clusters get mapped  to different IRRs \cite{siddique2018symmetry}.
 There is one representation (labeled $r=1$) which we call {trivial} and has dimension $d_1=\sum_{\alpha} L^\alpha$. All perturbations parallel to the synchronization manifold get mapped to this representation. Hence, the trivial representation is associated with all the clusters $\mathcal{C}^{1}_1,..., \mathcal{C}^{1}_{L^1}, \mathcal{C}^{2}_1,..., \mathcal{C}^{2}_{L^2},...$ However, it is possible that other IRR representations are only associated with some of the clusters (not all of them). 

We can now define the $\sum_{\alpha} N^\alpha n^\alpha$-dimensional orthonormal matrix 
${\tilde T}=\bigoplus_\alpha T^\alpha \otimes I_{n_\alpha}$.
Next, we will use the matrix $\tilde T$ to block-diagonalize Eq.\ \eqref{eq:stacked}.
Applying the transformation $ \text{\boldmath$\eta$} = {\tilde T}  \delta {\bf x}$ to \eqref{eq:stacked2} we obtain,

\begin{equation} \label{eq:stacked3}
\begin{split}
      \dot{\text{\boldmath$\eta$}} = & \left(
      \begin{bmatrix}
\sum_{u} J_{u}^{\alpha} \otimes DF^\alpha({\bf q}_u^\alpha)  & 0 &  \cdots\\
0 & \sum_{u} J_{u}^{\beta} \otimes DF^\beta({\bf q}_u^\beta) &  \cdots\\
\vdots & \vdots &  \ddots
\end{bmatrix} \right. \\
\
 + & \left. \begin{bmatrix}
  \sum_u B^{\alpha\alpha} J^\alpha_{u}  \otimes D\hat{H}^{\alpha \alpha}({\bf q}_u^{\alpha})  & \sum_u B^{\alpha\beta}   J^\beta_{u}  \otimes D\hat{H}^{\alpha \beta}({\bf q}_u^{\beta}) &  \cdots\\
\sum_u B^{\beta\alpha} J^\alpha_{u}  \otimes D\hat{H}^{\beta \alpha}({\bf q}_u^{\beta}) &  \sum_u B^{\beta\beta} J^\beta_{u}  \otimes D\hat{H}^{\beta \beta}({\bf q}_u^{\beta}) &  \cdots\\
\vdots & \vdots &  \ddots
\end{bmatrix} \right)
  \text{\boldmath$\eta$}, 
\end{split}
\end{equation}
where each transformed indicator matrix $J_u^\alpha= T^\alpha E_u^\alpha {T^{\alpha}}^T$ and each block $B^{\alpha \beta}={T^{\alpha}} A^{\alpha \beta} {T^{\beta}}^T.$

The advantage of \eqref{eq:stacked3} over \eqref{eq:stacked2} lies in the block-diagonal structure of the matrix $B=\oplus_{r=1}^R I_{d_r} \otimes \hat{B}_r$. 
The block $\hat{B}_1$ is associated with motion along the synchronization manifold. The blocks $\hat{B}_2,...,\hat{B}_R$ describe the dynamics transverse to the synchronization manifold.
As a result, we have decoupled the dynamics along the synchronous manifold from that transverse to it \cite{pecora2014cluster}. 
Moreover, each transverse block $r=2,..,R$ is associated with either an individual cluster or a subset of intertwined clusters \cite{pecora2014cluster}.
Thus the problem of studying the behavior of a perturbation away from the synchronous solution is typically reduced into many smaller problems, which can be analyzed independently one from the other.



{Next we discuss how Dependent and Independent Layer Symmetries affect the stability analysis.}
We recall that the matrix $B= T A T^{-1}$ can be written as a direct sum of blocks $\oplus_{r=1}^R I_{d_r} \otimes \hat{B}_r$. 
As each row of the matrix $T$ is associated to a specific cluster \cite{klickstein2019symmetry}, each one of the blocks $\hat{B}_r$ corresponds to a set of clusters which are identified by the rows of the matrix $T$.
 The trivial representation ($r=1$) is associated with all the clusters $\mathcal{C}^{1}_1,..., \mathcal{C}^{1}_{L^1}, \mathcal{C}^{2}_1,..., \mathcal{C}^{2}_{L^2},...$ The corresponding rows of the matrix $T$ are
 the eigenvectors associated to the eigenvalue 1 of the trivial representation; these have nonzero components all of the same sign. Now we look at the remaining `transverse' blocks of the matrix $B$ ($r>1$). The corresponding rows of the matrix $T$ are such that the sums of their entries is equal zero and are called {symmetry breakings}: each row, in fact, describes how a cluster, generated by a symmetry, may break into smaller ones. The transverse blocks can be divided into ILS blocks if the corresponding  symmetry breakings are all from the same layer and DLS blocks if the corresponding symmetry breakings are from different layers. We can then write $T=[T_{\text{ SYNC}}^T,T_{\text{ ILS}}^T,T_{\text{DLS}}^T]^T$ and the transformed vector $\text{\boldmath$\eta$}=[ \text{\boldmath$\eta$}_{\text{SYNC}}^T, \text{\boldmath$\eta$}_{\text{ILS}}^T, \text{\boldmath$\eta$}_{\text{DLS}}^T]^T$.

{The ILS blocks are symmetry breakings generated by the ILS subgroup.} From our definition of the ILS subgroup, each $\mathcal{K}_1^\alpha$ is a normal subgroup of $\mathcal{H}^\alpha$. 
Clifford theorem \cite{clifford1937representations} states that each IRR of $\mathcal{H}^\alpha$, when restricted on $\mathcal{K}_1^\alpha$, is either itself an IRR of $\mathcal{K}_1^\alpha$ or breaks up into a direct sum of IRRs of $\mathcal{K}_1^\alpha$ of the same dimension. The rows of the change of coordinates $T_{\text{ILS}}^\alpha$ (for layer $\alpha$) to the IRR of $\mathcal{K}_1^\alpha$, which are associated with the symmetry breaking perturbations of the ILS, are generated by the eigenvectors associated to the eigenvalue 1 of the projectors on the IRRs of $\mathcal{K}_1^\alpha$. 
These rows generate invariant subspaces of minimal dimension, and must therefore be rows of $T^\alpha$. 

We now consider for simplicity the case of a network with two layers, labeled $\alpha$ and $\beta$. In general, the transverse IRRs associated with the ILS subgroup in layer $\alpha$ will have  structure,
\begin{equation} \label{eta_ILS}
\dot{\text{\boldmath$\eta$}}_{{\text{ILS}}}=\Bigl[ \sum_{u} J_{u}^{\alpha} \otimes DF^\alpha({\bf q}_u^\alpha)+\sum_u B^{\alpha\alpha} J^\alpha_{u}  \otimes D\hat{H}^{\alpha \alpha}({\bf q}_u^{\alpha}) \Bigr] \text{\boldmath$\eta$}_{{\text{ILS}}}.
\end{equation}
Note the perturbation $\text{\boldmath$\eta$}_{{\text{ILS}}}$ is independent of the dynamics on the $\beta$ layer, i.e., of both $DF^\beta({\bf q}_u^\beta)$ and $D\hat{H}^{\beta \beta}({\bf q}_u^\beta)$. This indicates that the stability of ILS symmetries can be studied through a specific class of the Master Stability Function, which is the same as for single-layer networks \cite{pecora2014cluster}.
The transverse IRRs associated with the ILS subgroup in  layer $\beta$ will have an analogous structure as \eqref{eta_ILS}. On the other hand, the remaining transverse IRRs will have structure,
\begin{equation} \label{eta_DLS}
\begin{split}
      \dot{\text{\boldmath$\eta$}}_{{\text{DLS}}} = & \left(
      \begin{bmatrix}
\sum_{u} J_{u}^{\alpha} \otimes DF^\alpha({\bf q}_u^\alpha)  & 0 \\
0 & \sum_{u} J_{u}^{\beta} \otimes DF^\beta({\bf q}_u^\beta) 
\end{bmatrix} \right. \\
\
 + & \left. \begin{bmatrix}
  \sum_u B^{\alpha\alpha} J^\alpha_{u}  \otimes D\hat{H}^{\alpha \alpha}({\bf q}_u^{\alpha})  & \sum_u B^{\alpha\beta}   J^\beta_{u}  \otimes D\hat{H}^{\alpha \beta}({\bf q}_u^{\beta}) \\
\sum_u B^{\beta\alpha} J^\alpha_{u}  \otimes D\hat{H}^{\beta \alpha}({\bf q}_u^{\beta}) &  \sum_u B^{\beta\beta} J^\beta_{u}  \otimes D\hat{H}^{\beta \beta}({\bf q}_u^{\beta}) \\
\end{bmatrix} \right)
  \text{\boldmath$\eta$}_{{\text{DLS}}}. 
\end{split}
\end{equation}
 Note that the perturbations $\text{\boldmath$\eta$}_{{\text{DLS}}}$ depend on the dynamics of the systems in both layers $\alpha$ and $\beta$, through the mixed blocks appearing in Eq.\ \eqref{eta_DLS}. The structure of Eq. \eqref{eta_DLS} is substantially different than that of Eq. \eqref{eta_ILS},  which may lead to dramatic effects in terms of stability.

Consider now the multilalyer network in Fig.\ 1. The matrix $T$ is equal to 
\begin{equation} \label{eq:example_T}
T=\left[\begin{array}{cccccccc}
    0.5 &   0.5  &        0.5 &   0.5  & 0 &  0    &     0     &    0  \\
    0  &       0 &   0   &       0 &         1   &  0  & 0  &       0    \\
    0 & 0 & 0 & 0 & 0  &  \frac{1}{\sqrt{2}}  & 0 & \frac{1}{\sqrt{2}} \\
    0 & 0 & 0 & 0 & 0 & 0 & 1  & 0 \\
    \hline
    0.5 & -0.5 & 0.5 & -0.5 & 0 & 0 &0 & 0  \\
    0.5  &  -0.5 &   -0.5  &  0.5    &     0 &       0  &       0    &     0  \\
    \hline
    0.5  &  0.5  & -0.5 &  -0.5  & 0 & 0 & 0 & 0\\
    0  &  0 & 0 &  0 & 0 & \frac{1}{\sqrt{2}}  & 0 & -\frac{1}{\sqrt{2}}  
\end{array}\right] \hspace*{-6mm}
\begin{array}{l}
\left.
\begin{array}{l}\\\\\\\\\end{array}\right\} \text{SYNC}\\
\left.
\begin{array}{l}\\\\\end{array}\right\} \text{ILS breakings}\\
\left.
\begin{array}{l}\\\\\end{array}\right\} \text{DLS breakings}
\end{array}
\end{equation}


 
\normalsize 
 

The matrix $B=T A T^T$ is equal to:
\begin{equation} \label{Eq:Bfig1}
B = \left[\begin{array}{cccc|cc|cc}
    2 & 2 & \sqrt{2}  &  0 &  0  & 0  & 0 & 0\\
     2 & 0 & 0 &   0 &  0  &  0  & 0  & 0\\
     \sqrt{2} & 0 & 0 & \sqrt{2} & 0 & 0 & 0 & 0\\
     0 & 0 & \sqrt{2} & 0 &  0 & 0 & 0  & 0\\ \hline
    0 & 0 &  0 &  0 &  -2 &  0 &  0 &     0  \\
    0  &  0 & 0 &  0 &  0 &  0 &  0 & 0 \\ \hline
    0  &  0 & 0 &  0 &  0 &  0 & 0 & \sqrt{2}  \\
    0  &  0 & 0 &  0 &  0  &  0  & \sqrt{2} & 0  \\
\end{array}\right] \hspace*{-6mm}
\begin{array}{l}
\left.
\begin{array}{l}\\\\\\\\\end{array}\right\} \text{SYNC block}\\
\left.
\begin{array}{l}\\\\\end{array}\right\} \text{ILS blocks}\\
\left.
\begin{array}{l}\\\\\end{array}\right\} \text{DLS block}
\end{array}
\end{equation}
As can be seen, the matrix $B$ is the direct sum of one SYNC block, two independent ILS blocks (corresponding to breakings in the $\alpha$ layer which are independent of the $\beta$ layer), and one DLS block (corresponding to breakings in the $\alpha$ layer and in the $\beta$ layer which are dependent on each other.) The transverse ILS and DLS blocks describe linear (local) stability of clusters, when the dynamics is linearized about the quotient network time evolution. In what follows we will pay particular attention to stability of DLS blocks, which are a specific feature of multilayer networks.

 To provide analytical insight, we first consider a two-layer network described by the following discrete-time dynamics,
 
 \begin{equation} \label{eq:discrtime}
 \begin{array}{rcl}
     {\bf{x}^\alpha}^{k+1} &=& \mbox{rem}(c^\alpha {\bf{x}^\alpha}^{k} + \sigma A^{\alpha \alpha} {\bf{x}^\alpha}^{k} + \sigma A^{\alpha \beta} {\bf{x}^\beta}^{k})
\\
     {\bf{x}^\beta}^{k+1} &=& \mbox{rem}(c^\beta {\bf{x}^\beta}^{k} + \sigma A^{\beta \beta} {\bf{x}^\beta}^{k} + \sigma A^{\beta \alpha} {\bf{x}^\alpha}^{k}),
\end{array}
\end{equation}
where the vectorial function $\mbox{rem}({\bf{x}})$ returns a vector whose entries are the remainder of the integer division of the entries of the vector ${\bf{x}}$ by the scalar $1$ and $c^\alpha$ and  $c^\beta$ are tunable layer-specific  scalar parameters. 

Stability is described by the following set of equations,
\begin{equation}
    \begin{array}{rcl}
     \delta {\bf{x}^\alpha}^{k+1} &=& c^\alpha {\delta \bf{x}^\alpha}^{k} + \sigma A^{\alpha \alpha} {\delta \bf{x}^\alpha}^{k} + \sigma A^{\alpha \beta} {\delta \bf{x}^\beta}^{k}
\\
     \delta {\bf{x}^\beta}^{k+1} &=& c^\beta {\delta \bf{x}^\beta}^{k} + \sigma A^{\beta \beta} {\delta \bf{x}^\beta}^{k} + \sigma A^{\beta \alpha} {\delta \bf{x}^\alpha}^{k}.
\end{array}
\end{equation}

We now consider a generic DLS block of the form $\begin{pmatrix} a & b \\ b & a \end{pmatrix}$ ($a=0$ and $b=\sqrt{2}$ in the DLS block of Eq.\ \eqref{Eq:Bfig1}), to which corresponds a perturbation of the form,
\begin{equation} \label{eq:eta_DLS} 
      {\text{\boldmath$\eta$}}_{{\text{DLS}}}^{k+1} =  \left(
      \begin{matrix}
 c^\alpha+  \sigma a   &  \sigma b   \\
 \sigma b  & c^\beta+  \sigma a \\
\end{matrix} \right)
  \text{\boldmath$\eta$}_{{\text{DLS}}}^k, 
\end{equation}
with $c^\alpha \neq c^\beta $ and the case $c^\alpha=c^\beta$ corresponding to an ILS perturbation. The eigenvalues have the following expression,
\begin{equation} \label{rho}
    \rho^{\pm}=\bar{c} + a \sigma \pm  (b^2 \sigma^2 + \delta_c^2)^{1/2},
\end{equation}
where $\bar{c}=(c^\beta+c^\alpha)/2$ is the average layer-specific parameter and $\delta_c=(c^\beta-c^\alpha)/2$ measures how different are the systems in the two layers (with $\delta_c=0$ corresponding to the case of an ILS.) From this equation we see the larger (smaller) eigenvalue is an increasing (decreasing) function of $\delta_c$, and it follows that the best condition in terms of stability is achieved for $\delta_c=0$. We thus conclude that DLS perturbations are more difficult to stabilize than ILS perturbations.

We present here a conjecture expanding on the above conclusion: that DLS clusters are generally more difficult to stabilize compared to the case in which the systems in different layers are of the same type.  We base this on the following reasoning.  Consider an ideal situation in which
 at first the parameters of the systems in the different layers are identical and then they are increasingly perturbed to take on different values in different layers. The Gershgorin circle's theorem states that each eigenvalue of a matrix lies within at least one of the Gershgorin discs centered at the entries on the main diagonal and having radius equal to the sum of the off diagonal entries. As perturbing the individual system's parameters corresponds to varying the centers of the Gershgorin discs (but not the radius), we can expect the eigenvalues to become increasingly spread out as the systems are made increasingly different from one another, resulting in reduced stability.

The particular network in Fig.\ 1 has one transverse irreducible representation (DLS) in the form of Eq. \eqref{eq:eta_DLS},  with $a=0$ and $b=\sqrt{2}$, 
corresponding to simultaneous loss of synchronization between nodes $(1,3)$ in the $\beta$ layer and between the pairs of nodes $(1,2)$ and $(3,4)$ in the $\alpha$ layer. For this multilayer network, we consider the dynamics of Eq.\ \eqref{eq:discrtime}  and study the effects of changing the parameters $c^\alpha$ and $c^\beta$ on stability. For the maps described by Eq.\ \eqref{eq:discrtime}, as well as in the following for other kind of systems, we measure the pairwise synchronization error, defined as,
\begin{equation}
    E^\alpha_{ij}=<\|{\bf{x}}^\alpha_i-{\bf{x}}^\alpha_j\|>,
\end{equation}
between nodes $i$ and $j$ in layer $\alpha=1,...,M$, where the symbol $<...>$ indicates a temporal average, computed after the transient has elapsed.

Panel {\textbf a} of Fig.\ \ref{fig:rho} shows $E_{13}^\alpha$ and $E_{13}^\beta$ vs the parameter $\delta_c$. $\delta_c=0$ corresponds to identical systems in the two layers, increasing values of $\delta_c$ indicate the systems in the two layers are increasingly different. Synchronization is simultaneously lost in both layers for $\delta_c \gtrapprox 0.35$. Panel {\textbf b} is a plot of the eigenvalues $\rho^+$ and $\rho^-$ in Eq.\ \eqref{rho} as a function of $\delta_c$, which shows that $\rho^+$ ($\rho^-$) increases (decreases) with $\delta_c$. Loss of stability occurs when either one of the two eigenvalues is either larger than $1$ or smaller than $-1$. It is thus expected that stability may be lost for increasing values of $\delta_c$, i.e., as the systems in the two layers become increasingly different. From Panel {\textbf b} of Fig.\ \ref{fig:rho} we see that
$\rho^+$ grows larger than $1$ for $\delta_c \gtrapprox 0.35$. 


\begin{figure}[htbp!]
    \centering
    \includegraphics[width=15cm]{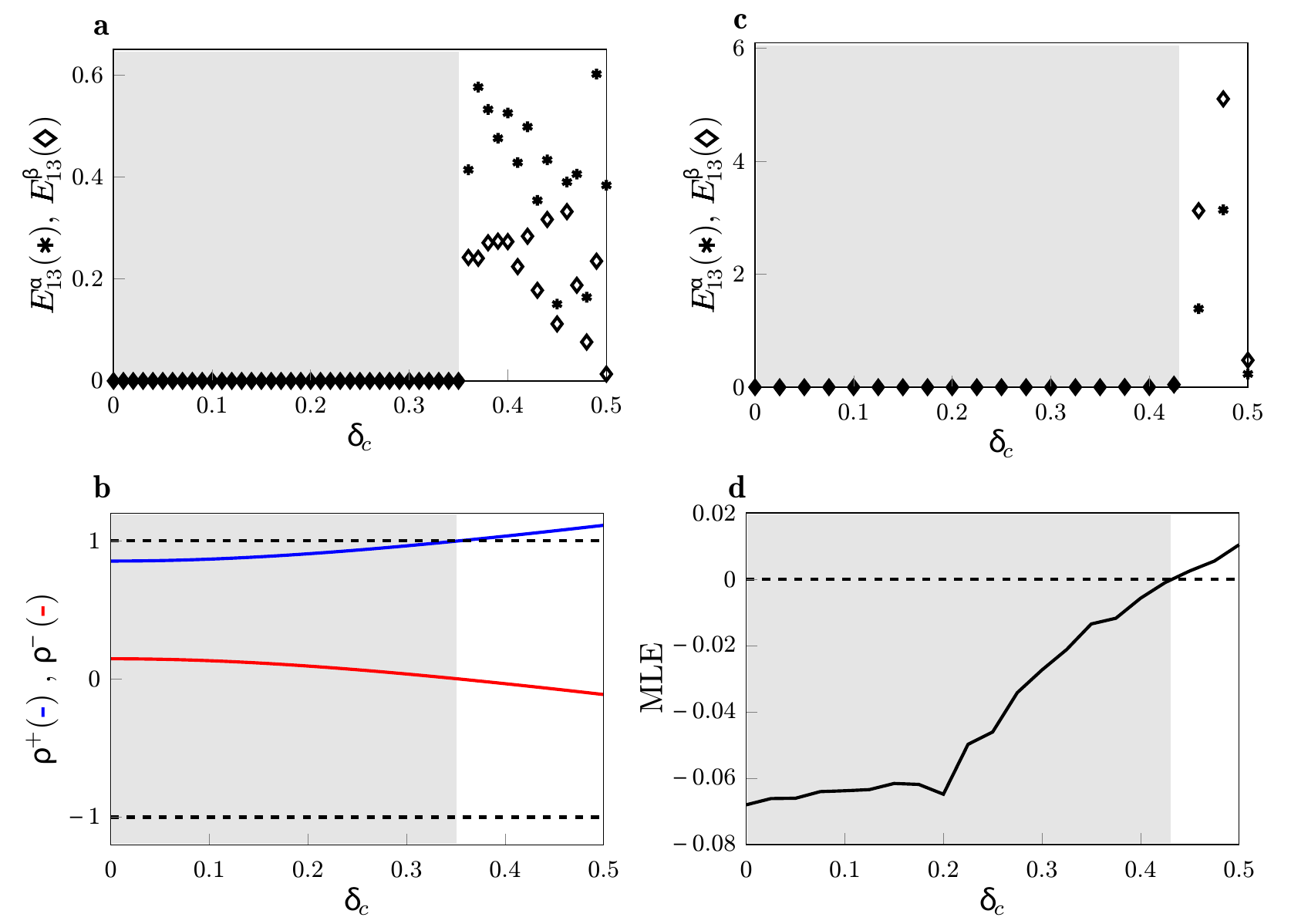}
    \caption{\textbf{Synchronization of DLS clusters.} {\textbf{a}} Discrete time maps. Synchronization errors $E_{13}^\alpha$ (asterisks) and $E_{13}^\beta$ (diamonds) vs the parameter $\delta_c$. $\delta_c=0$ corresponds to identical systems between the $\alpha$ and in the $\beta$ layers. Increasing values of $\delta_c$ are for increasingly different systems in the two layers. Synchronization is lost in both layers for $\delta_c \gtrapprox 0.35$. {\textbf{b}} Discrete time maps. Eigenvalues $\rho^+$ and $\rho^-$ in Eq. \eqref{rho}. As can be seen, $\rho^+$ ($\rho^-$) increases (decreases) with $\delta_c$. Loss of stability occurs when either one of the two eigenvalues is either larger than $1$ or smaller than $-1$. $\rho^+$ becomes larger than $1$ for $\delta_c \gtrapprox 0.35$. {\textbf{c}} Van Der Pol oscillators. Synchronization errors $E_{13}^\alpha$ (asterisks) and $E_{13}^\beta$ (diamonds) vs the parameter $\delta_c$. $\delta_c=0$ corresponds to identical systems between the $\alpha$ and in the $\beta$ layers. Increasing values of $\delta_c$ are for increasingly different systems in the two layers. Synchronization is lost in both layers for $\delta_c \gtrapprox 0.4$. {\textbf{d}} Van Der Pol oscillators. Maximum Lyapunov Exponent of the DLS block \eqref{eta_DLS} vs the parameter $\delta_c$. }
    \label{fig:rho}
\end{figure}

We then consider the case of Eqs.\ \eqref{eq:gen_dyn} with $M=2$ layers, $A^{\alpha \alpha}, A^{\beta \beta}, A^{\alpha \beta}=(A^{\beta \alpha})^T$ corresponding to the multilayer network in Fig.\ 1, $n^\alpha=n^\beta=2$, ${\bf x}^\alpha=[{x}^\alpha, {y}^\alpha]$, ${\bf x}^\beta=[{x}^\beta, { y}^\beta]$,
\begin{equation}
{\bf F}^\alpha({\bf x}^\alpha)=\left[\begin{array}{c}
     y^\alpha  \\
     -x^\alpha-0.2 y^\alpha [(x^\alpha)^2-c^\alpha]
\end{array}\right], \quad \quad 
{\bf F}^\beta({\bf x}^\beta)=\left[\begin{array}{c}
     y^\beta  \\
     -x^\beta-0.2 y^\beta [(x^\beta)^2-c^\beta] 
\end{array} \right],
\end{equation}
which both correspond to the dynamics of the Van der Pol oscillator, with layer-dependent parameters $c^\alpha$ and $c^\beta$. Moreover, we set
\begin{equation}
{\bf H}^{\alpha \beta}({\bf x})=\begin{pmatrix}
-1 & 0\\ 0 & 0
\end{pmatrix} {\bf x},
\end{equation}
for all pairs $\alpha,\beta=1,2$ and $\sigma^{\alpha \beta}=0.15$ for all pairs $\alpha,\beta=1,2$. We set the layer-dependent parameters $c^\alpha=(1+\delta_c)$ and $c^\beta=(1-\delta_c)$, so that $\delta_c=0$ corresponds to identical systems in the two layers and increasing values of $\delta_c$ indicate the systems in the two layers are increasingly different. We then numerically investigate stability of the DLS as the parameter $\delta_c$ is increased. This can be seen in panel {\textbf c} of Fig.\ \ref{fig:rho}, which shows that synchronization between oscillators $1$ and $3$ from the $\alpha$ layer and synchronization between oscillators $1$ and $3$ from the $\beta$ layer is simultaneously lost for $\delta_c \gtrapprox 0.4$. Panel {\textbf d} of Fig.\ 2 shows the numerically computed Maximum Lyapunov exponent of the DLS block \eqref{eta_DLS}. We also note that throughout the whole $\delta_c$ interval considered, neither nodes $1$ and $2$, nor nodes $1$ and $4$ from the $\alpha$ layer ever synchronize (not shown). Overall, Fig.\ \ref{fig:rho} shows that increased heterogeneity between the nodes in the two layers, can lead to loss of DLS stability. Further numerical evidence of this is presented for the cases of the Lorenz oscillator and of the Roessler oscillator in Supplementary Note 4. Supplementary Note 5 investigates how varying the intra-layer and inter-layer coupling strengths affects CS stability.

\subsection*{An experimental testbed for cluster synchronization}\label{sec: experiment}
We apply the techniques from the previous sections on an experimental testbed circuit. 
We implement three different kinds of electronic oscillators, using widely available, affordable components that can be assembled on breadboards. 
Simplicity, low-cost, ease of fabrication and availability of a large volume of previous studies make electronic circuits an ideal testbed for multilayer network studies \cite{blaha2019cluster}.
The circuit is composed of three different kinds/layers of nodes: one linear resonator, which we call the ``jumper'', two FitzHugh-Nagumo (FHN) oscillators, and four Colpitts oscillators.
The jumper is a linear resonator (i.e., without input, its oscillations damp to zero) with two-dimensional governing equations;
the FHN is a relaxational nonlinear oscillator with two-dimensional governing equations~\cite{Keener_1983jy}; 
the Colpitts is a sinusoidal nonlinear oscillator with three-dimensional governing equations~\cite{Kennedy_1994}.
All three kinds of nodes have similar uncoupled frequencies. A full schematic of the experimental circuit is included in Figure \ref{fig:circuit}.
Figure \ref{fig:circuit} also states the measured parameter values of each component. 

{Figure  \ref{fig:circuit_mul} shows a simplified multilayer schematic of the circuit with the jumper as the $\alpha$ layer, the FHNs as the $\beta$ layer, and the Colpitts as the $\gamma$ layer.}
We couple the different oscillators in three ways. 
 On the Colpitts layer, we induce inductive coupling through mutual inductance, red in Fig.\  \ref{fig:circuit_mul}.  Between the Colpitts layer and the FHN layer we introduce resistive coupling; this couples each FHN circuit with two Colpitts, blue in Fig.\  \ref{fig:circuit_mul}.
 Between the FHN and jumper layer, we induce inductive coupling between each FHN and an inductor in the jumper, orange in Fig.\  \ref{fig:circuit_mul}.

\begin{figure}[h!]
    \centering
    \includegraphics[width=\textwidth]{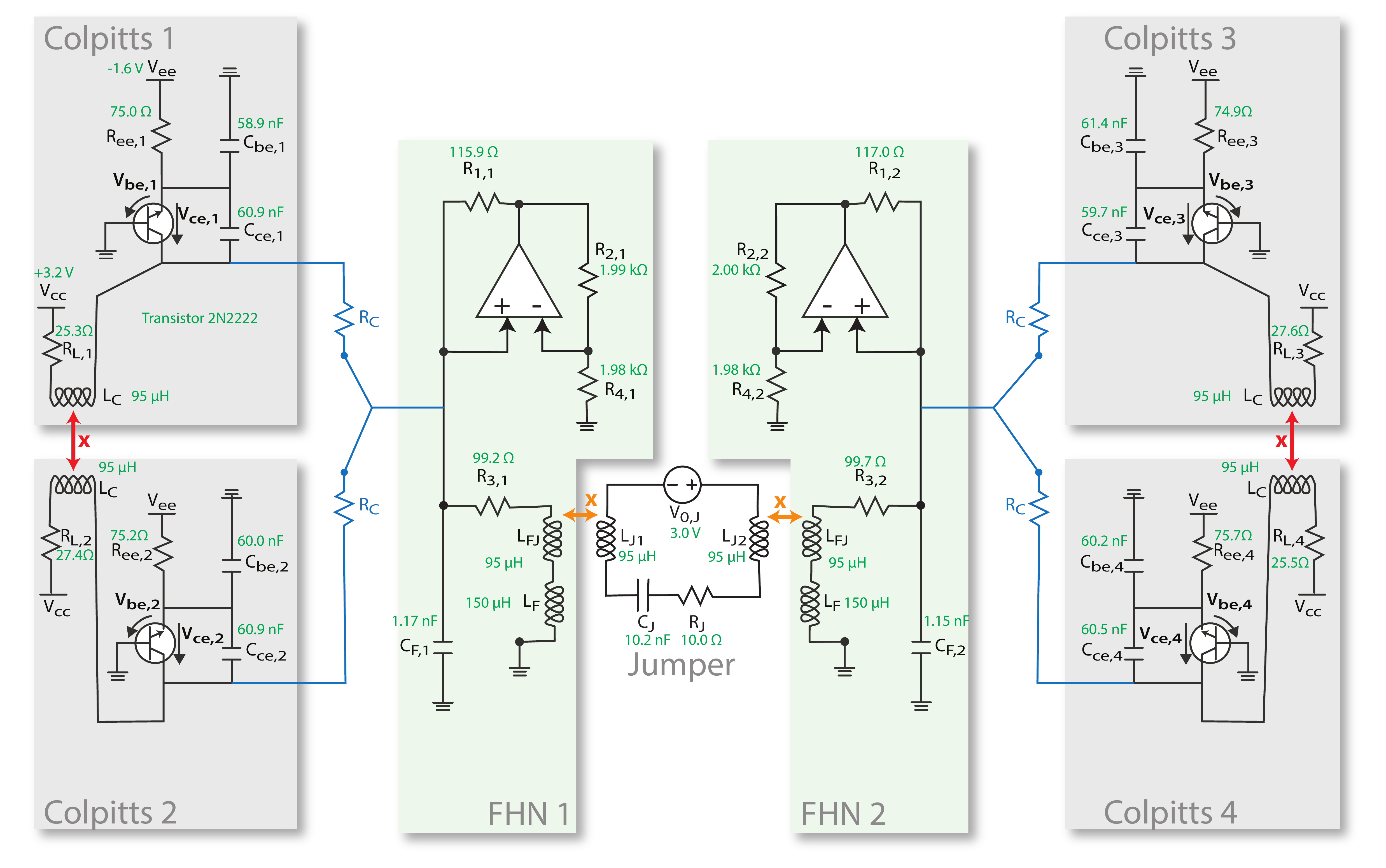}
    \caption{\textbf{Schematic of the experimental setup.} {$R_C=2.2k\Omega.$ We vary the magnetic coupling between Colpitts 1 and 2 and Colpitts 3 and 4, $k_C$, from $-0.4-0.4$ by varying separation $x$ (red). We hold the separation between the FHNs and the jumper  (x in orange) constant such that the coefficient of mutual inductance, $k_{FJ}$, is maintained equal to 0.35. }}
    \label{fig:circuit}
\end{figure}

\begin{figure}[hbtp!]
    \centering
    \includegraphics[width=0.8\textwidth]{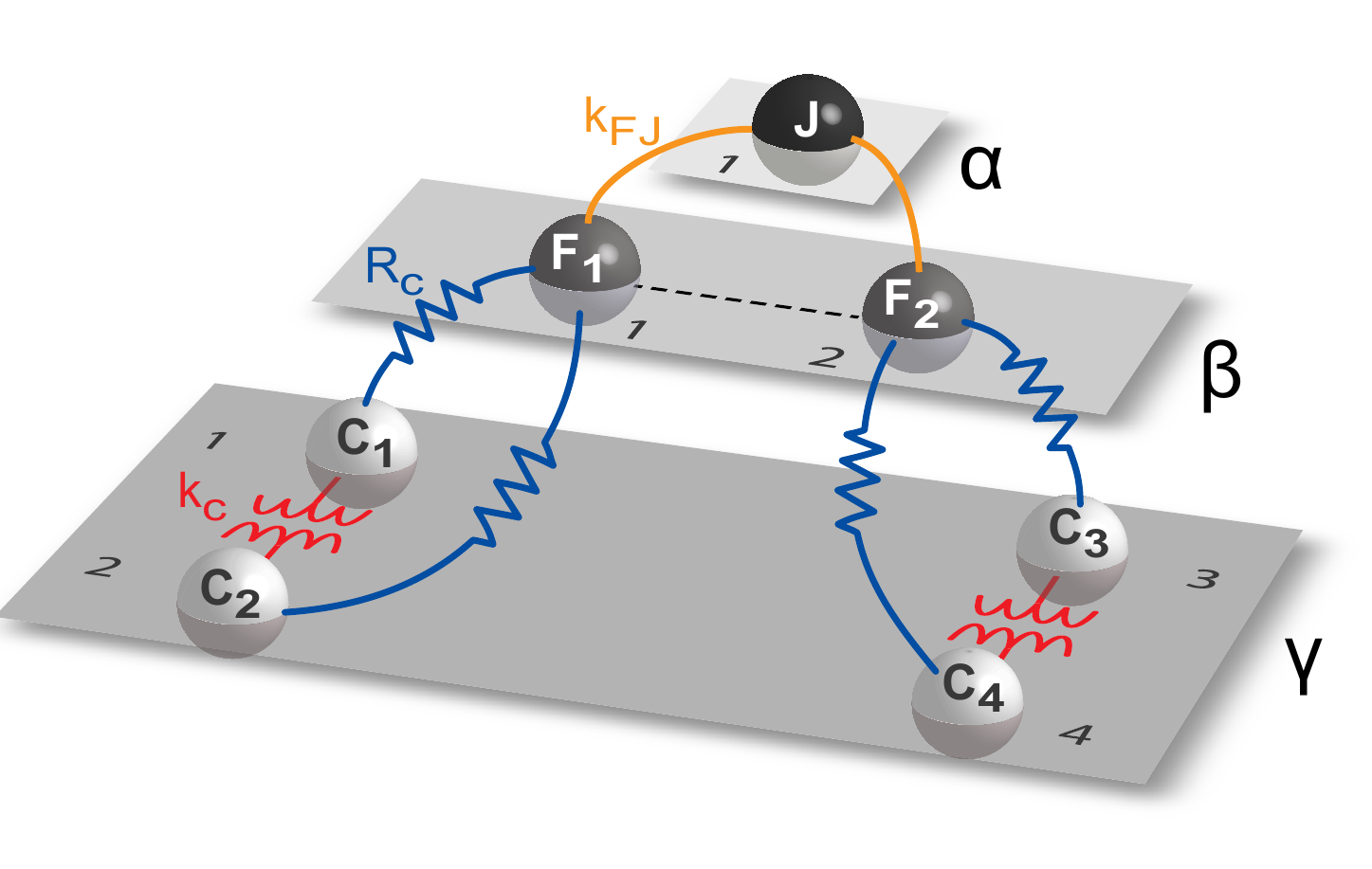}
    \caption{\textbf{Multilayer representation of the experimental circuit.} {Layer $\alpha$ contains the jumper node. Layer $\beta$ contains two FHN oscillators.  Layer $\gamma$ contains four Colpitts oscillators with inductive intralayer coupling shown in red. The interlayer resistive coupling between the $\beta$ and $\gamma$ layers is shown in blue. The interlayer inductive coupling between the $\alpha$ and $\beta$ layers is shown in orange; this coupling introduces a virtual FHN intralayer coupling which we discuss in the analysis of the system.}}
    \label{fig:circuit_mul}
\end{figure}


As shown in the above schematic and equations, we can describe the system with three layers composed as follows.
Layer $\alpha$ includes the jumper alone,
layer $\beta$ includes 2 FHNs and the intra-layer connections, with $\mathcal{G}^\beta=\{(),(1,2)\}$ and
layer $\gamma$ includes 4 Colpitts and one intra-layer connection that connects nodes 1 with 2, and nodes 3 with 4. The group of symmetries of this layer is
\begin{equation}
\mathcal{G}^\gamma=\{(),(1,2),(3,4),(1,2)(3,4),(1,3)(2,4),(1,4)(2,3)\}.
\end{equation}
We have the inter-layer connections:
\begin{equation}
A^{\alpha\beta}=\begin{bmatrix}1 & 1\end{bmatrix} = (A^{\beta\alpha})^T,\qquad
A^{\beta\gamma}=\begin{bmatrix}1 & 1 & 0 & 0 \\ 0 & 0 & 1& 1\end{bmatrix} = (A^{\gamma\beta})^T.
\end{equation}
All the permutations in the groups $\mathcal{G}^\alpha,\ \mathcal{G}^\beta,\ \mathcal{G}^\gamma$ are compatible, then
$\mathcal{H}^\alpha=\mathcal{G}^\alpha$, $\mathcal{H}^\beta=\mathcal{G}^\beta$, $\mathcal{H}^\gamma=\mathcal{G}^\gamma.$

The interlayer connections define the conjugate classes
\begin{equation}
\mathcal{K}^\alpha=();\;
\mathcal{K}_1^\beta=(),\ \mathcal{K}_2^\beta=(1,2);\; \mathcal{K}_1^\gamma=\{(),(1,2),(3,4),(1,2)(3,4)\},\  \mathcal{K}_2^\gamma=\{(1,3)(2,4),(1,4)(2,3)\}.
\end{equation}
This set defines the group of symmetry of the multilayer network.
There are three possible clustered patterns: (1a) layer $\beta$ and layer $\gamma$ fully synchronized,  (1b) layer $\beta$ fully synchronized and layer $\gamma$ clustered synchronized (either 1 with 3 and 2 with 4 or, equivalently, 1 with 4 and 2 with 3), and (2) layer $\beta$ not synchronized and layer $\gamma$ synchronized in clusters (1 with 2 and 3 with 4). The stability of each one of these clustered patterns is analyzed in detail in the Supplementary Note 6 and is here summarized in Figure~\ref{fig:patt}. Plot {\textbf a} shows the two MLEs transverse to solution 1a. Solid and dotted lines refer to transverse blocks $B^2$ and $B^{1b}$, respectively (blocks defined in the Supplementary Note 6).  Pattern 1a is stable when both the curves are negative. {\textbf b} shows MLEs transverse to stable solutions on the synchronized manifold 1b. Blue curves refer to the transverse MLE of solution of kind 1a, while the red curve refers to the transverse MLE of solution of kind 1b.  {\textbf c} shows MLEs transverse to stable solutions on the synchronized manifold 2. Blue curves refer to the transverse MLE of solution of kind 1a, while the yellow curves refer to the transverse MLE of solution of kind 2. 
    
  {
The equivalence class subgroup $\mathcal{K}_1^\gamma$ shows the ILS's for the $\gamma$ layer. The symmetry breaking of the $(1,2)(3,4)$ permutation for the $[1,2,3,4]$ cluster causes the cluster to break into two clusters, $[1,3],[2,4]$ or $[1,4],[2,3]$. Those ILS permutations are what allows layer $\beta$ to remain synchronized in case (1b) above. Note, however, that breaking of the individual permutations like (1,2) by itself will cause a breaking of one of the $\mathcal{K}_2^\gamma$ permutations, so it is essential to check that the symmetry breakings of $\mathcal{K}_1^\gamma$ are not inducing other symmetry breakings in the same layer, but from other cosets.  If they are not, then the other layers are unaffected by the loss of an ILS.
}

\begin{figure}[h!]
\centering
\includegraphics[width=1\textwidth]{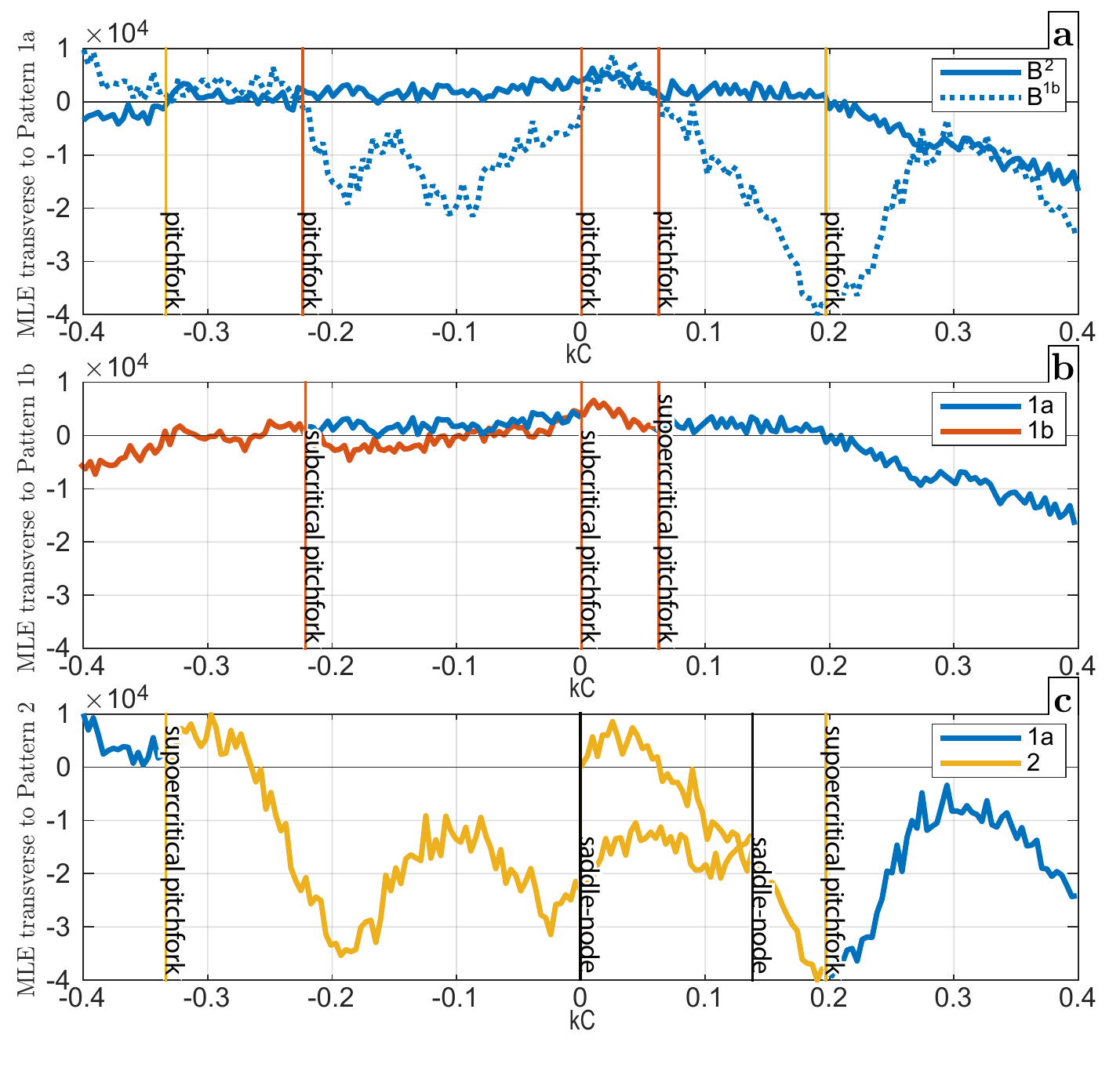}
\caption{\textbf{Stability of different cluster synchronization patterns.} {\textbf a} MLEs transverse to solution 1a. Solid and dotted lines refer to different transverse blocks. When a line crosses zero it identifies a symmetry breaking bifurcation in one of the other invariant manifold (red: bifurcations on CS manifold 1b - yellow: bifurcations on CS manifold 2). Pattern 1a is stable when both the curves are negative. {\textbf b} MLEs transverse to stable solutions on the synchronized manifold 1b. Blue curves refer to the transverse MLE of solution of kind 1a, while the red curve refers to the transverse MLE of solution of kind 1b. Vertical lines indicate the bifurcations of the quotient system. {\textbf c} MLEs transverse to stable solutions on the synchronized manifold 2. Blue curves refer to the transverse MLE of solution of kind 1a, while the yellow curves refer to the transverse MLE of solution of kind 2. Vertical lines indicate the bifurcations of the quotient system. Colored lines represent symmetry breaking bifurcation (inferred from {\textbf a}), while black lines are other bifurcations.}
\label{fig:patt}
\end{figure}

\begin{figure}[h!]
\centering
\includegraphics[width=0.88\textwidth]{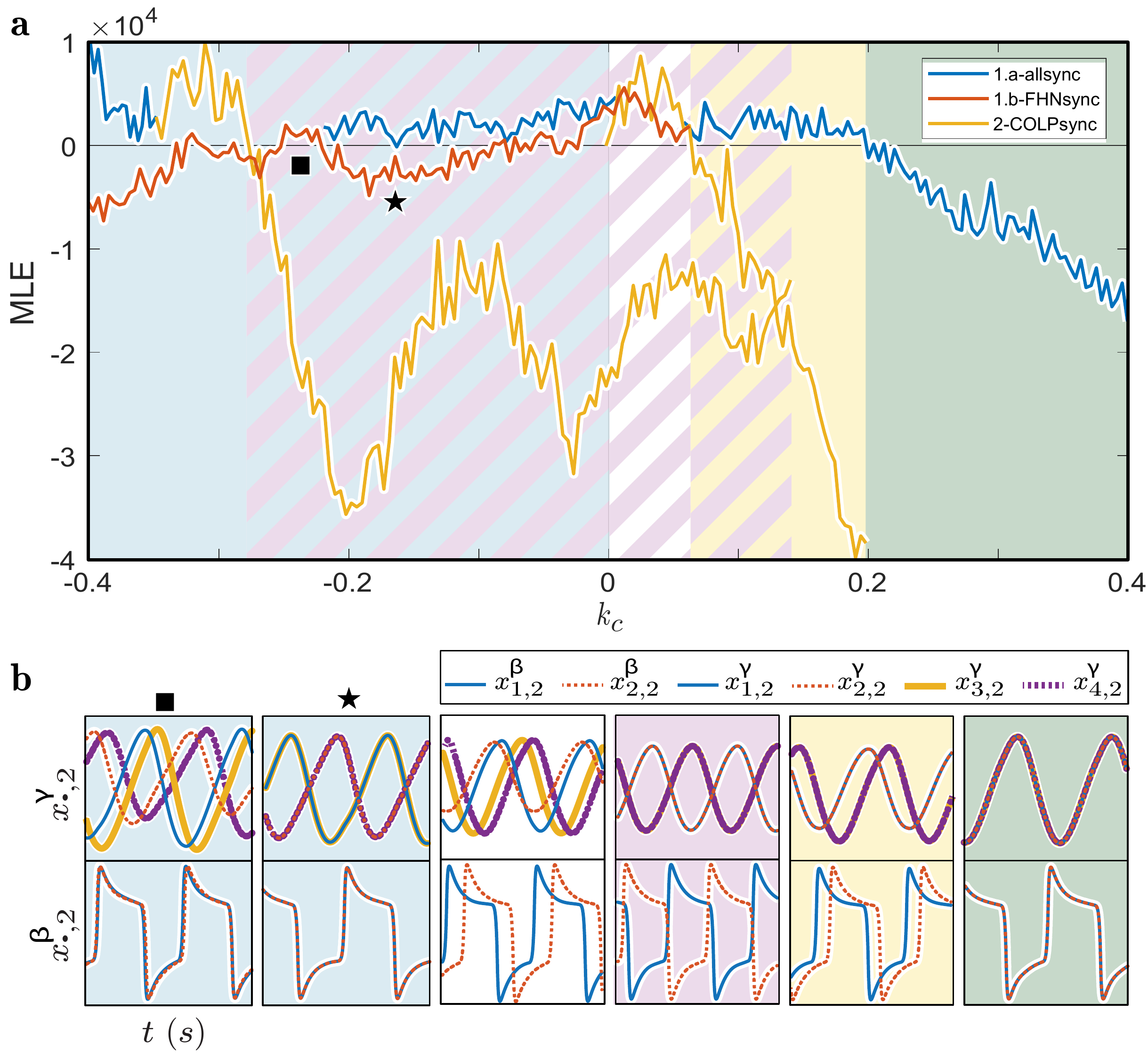}
\caption{\textbf{Onset of different cluster synchronization patterns as varying $k_C$.}
We observe five kinds of clustered behavior, which are denoted by the shading of the background.
Shading with two-colored striping indicates bistability between the two states represented by each color.
\textbf{a} MLE transverse to each identified clustered solution.
(Green) Pattern 1a;
(Yellow) Pattern 2-($\pi/2$);
(Lavender) Pattern 2-($\pi$);
(Blue) Pattern 1b;
and
(White) No pattern.
\textbf{b} different stable solutions present in each of the structurally different regions (identified by a bifurcation of a stable solution both in one of the quotient networks or transverse to them).
(Star) A stable solution of kind 1b occurs when the red line is negative;
(Square) a modified version of kind 1b occurs when the red line is positive.
}
\label{fig:summary}
\end{figure}

With all three patterns, we can draw Figure \ref{fig:summary}, where we report the maximum of the two blue curves in Figure \ref{fig:patt}{\textbf a} (blue), the red curve in Figure \ref{fig:patt}{\textbf b} (red) and the yellow curves in Figure \ref{fig:patt}{\textbf c} (yellow).
We can thus identify the parameter regions in which the different analyzed behaviors are present.
Below we comment on Figure~\ref{fig:summary} and explain the sequence of bifurcations that occur in the system as the parameter $k_C$ is decreased (i.e., from right to left looking at the figure.) For large positive coupling, 1a is stable and it is the only attractor of the system.
    This region is shaded in green in Figure~\ref{fig:summary}.
    At $k_C=0.2$, we observe the split of the two nodes in layer $\beta$ (the FHNs, transverse direction toward 2 in the $T$ matrix). The symmetry breaking (pitchfork) bifurcation is supercritical, and a solution of kind 2 is  born. This region is shaded in yellow.
     At $k_C=0.14$ a second stable clustered solution of type 2 is born through a saddle-node bifurcation (via the quotient network 2 dynamics). {This solution is characterized by a phase separation of the two FHN oscillators of $\pi$, and a slightly higher oscillations frequency.}
    It remains stable until $k_C$ is reduced to -0.28.
    We thus have a region with two solutions of kind 2: the one created at $k_C=0.2$ where the two FHNs are slightly phase shifted (whose MLE can be found looking at the yellow curve from right to left) and  the one created at $k_C=0.14$ where the
    two FHNs are anti-phase (whose MLEs can be found looking at the yellow curve from left to right).
    We call these two states pattern 2-($\pi/2$) and pattern 2-($\pi$), respectively.
    The state created at $k_C=0.14$ is shaded in lavender; the bistable region is striped lavender and yellow.
     At $k_C=0.07$, the clustered solution 2-($\pi/2$) loses its stability giving rise to an unsychronous solution, shaded in white. 
    This solution disappears through a saddle-node bifurcation at $k_C=0$. 
     For $k_C \in[-0.27,0]$, the system has two possible solutions. 
    The first (shaded lavender) is pattern 2-($\pi$). 
    The second (shaded in blue) has two possible behaviors: when the red curve is negative, it is of kind 1b (see star for example), or when the red curve is positive, it is a slightly non synchronized modification of pattern 1b (see square for example) in which the two FHNs remains locked, while the Colpitts are not perfectly synchronized. 
    This last solution is present because the inductively coupled Colpitts pairs are antiphase; 
    as a result, the FHN (which is coupled to both Colpitts) sees a very small net signal from the two Colpitts.
    Consequently, the Colpitts can very nearly synchronize even though they receive slightly different inputs.
    Finally, at $k_C=-0.28$ the pattern 2-($\pi$) solution undergoes a catastrophic symmetry breaking bifurcation, leaving 1b as the only possible attractor of the system (or its slight modification when the red curve is positive.) Experimental data showing qualitative agreement with the numerical results in  Figure~\ref{fig:summary} are included in the Supplemetary Note 7.

\section*{Discussion}\label{sec: discussion}

In this paper we study analytically, numerically, and experimentally the general problem of cluster synchronization (CS) in multilayer networks; {these systems are composed of heterogeneous components that may interact in multiple ways}. 
We first present a general set of equations that describes the dynamics of a multilayer network. 
We then define for the first time the group of symmetries of a multilayer network and explain how to compute it. An analysis of several datasets of real multilayer networks shows they possess large number of symmetries.

{We investigate the stability of the CS patterns which correspond to the orbits of the multilayer network.
The symmetries of each layer as well as the particular pattern of interlayer connections determine the CS patterns in multilayer networks; 
we consequently describe how each layer's symmetry group relates to the symmetry groups of other layers through interlayer connections.  
The interplay between the intralayer and interlayer couplings enables distinction between symmetries that affect the types of allowed dynamics in the other layers (DLS) and those that do not (ILS.)
The latter form a normal subgroup of the full symmetry group.
In particular, a symmetry-breaking bifurcation, associated  with the ILS subgroup, breaks up a clustered pattern in one layer of the multilayer network without altering the possible clusters in the other layers. }


With an IRR change of coordinates we decouple the stability problem into several simpler (lower dimensional) problems. 
First, we decouple perturbations along directions parallel to the synchronization manifold from those transverse to it; 
the latter determine stability of the clustered motion. 
Second, we decouple the equations for the transverse perturbations into several independent blocks; 
each block corresponds to the stability of either an individual cluster or a set of intertwined clusters \cite{pecora2014cluster}. 
When two or more clusters (which may belong to different layers) are intertwined, they are either all stable or all unstable. {We see that Dependent Layer Symmetries yield blocks of a different structure than those arising in the study of single layer networks, which we show has a profound effect on the stability of the clusters involving these symmetries. In particular, we show analytically for a specific class of networks that DLS clusters are more difficult to synchronize than in the case in which the systems in different layers are of the same type, which is also confirmed numerically in simulations involving multilayer networks of Van der Pol, Lorenz, and Roessler oscillators (Supplementary Note 5.)} 


{We performe experiments with a fully analog multilayer network with seven electronic oscillators of three different kinds coupled with two kinds of coupling.
{The testbed circuit has many features that occur in natural multilayer systems, like noise and parameter mismatches. The experiment is a good test of the theoretical framework because it allow us to understand how theoretical predictions, made with simplifying assumptions, can be a guide to better understand real phenomena. In fact, }
the experimental results largely match the theoretical predictions; we observe all the predicted cluster states in  the right part of the parameter space. 
}
{With respect to the experimental realization, the model includes several simplifications: 
 we assume all the nodes within the same layer are identical,  we assume that all the connections of the same type are identical,
 we use a simplified FHN model which neglects two resistors, 
 we use ideal models for the operational amplifiers and the transistors,
and we assume no noise and neglect any stray inductance, capacitance, or resistance.
We discuss why the experimental and theoretical results differ in our seven node electronic system in Supplementary Note 8. 
A broader, multisystem analysis of the robustness of the method to heterogeneity, noise, and network nonideality would enhance the utility of the method.}

A variety of real world systems are multilayer networks that can exhibit clustering. 
Dynamical situations, to which our analysis may be relevant include: opinion dynamics and consensus among individuals interacting through different communication systems, for which clustering may show up {on average}, see e.g., \cite{sorrentino2019symmetries},
 the dynamics of central pattern generators,  small networks of
similar neurons, which might show symmetries and clustering since 
synchronization is part of their dynamics, see e.g., \cite{lodi2018design}, and  
 electronic networks and in general man-made systems formed of many identical 
subsystems or nodes, which may produce cluster synchronization. 
{Recent work \cite{ishizaki2018graph} has studied how network symmetries may affect synchronization modes of power grids  and even suggested that symmetries may enhance complete synchronization in multilayer grids, characterized by the presence of different energy sources, such as power generators, wind turbines, solar, etc.} 
{Our work describes how clusters may arise in multilayer networks with a given structure; with further study, it may be possible to infer the structure of a multilayer network given the observation of clusters.
}

   {A main limitation of our approach is its scalability with the number of nodes of a multilayer network.  While the size of the network for which symmetries can be found is very large \cite{macarthur2008symmetry},   the size for which the stability analysis can be performed is typically much smaller \cite{cho2017stable}. We also model the systems in each layer as exactly identical; this is not a characteristic of experimental systems (for a discussion of the discrepancies between our experiment and our simulations see Supplementary Note 8.)
   A recent paper analyzed cluster synchronization in the presence of nearly identical systems \cite{sorrentino2016approximate}, though the approach of \cite{sorrentino2016approximate} is not easily generalizable  to the case of multilayer networks.}
   {
   Although our multilayer framework does not allow node heterogeneity and noise, we saw broad agreement between our predictions and  experimental circuit behavior; 
   this circuit contained slightly heterogeneous nodes with noise. Within this study, we cannot quantitatively state when noise or heterogeneity will qualitatively impact our predictions, but the experiment demonstrates that there is some tolerance. Further work will be needed to characterize these limitations.}
   
   {Finally, despite the generality of the theory proposed in this manuscript, several extensions are possible. An important direction for the research is to allow both intra-layer and inter-layer connections to be directed. Another direction is to study the formation of CS patterns that are not related to symmetries  in multilayer networks, see e.g. \cite{siddique2018symmetry}. {The case of group consensus which can be seen as a very special case (i.e., with linear dynamics) of the cluster synchronization problem considered here, has recently been studied in  \cite{sorrentino2020group}.}}

\section*{Data availability}
All data supporting the findings of this study are available within the article and its supplementary information files and from the corresponding author upon reasonable request.

\section*{Methods}

\subsection*{Formal Proofs}

\begin{theorem}
\label{subgroup}
$\Ha$ is a subgroup of $\Ga$ and $\Hb$ is a subgroup of $\Gb$.
\end{theorem}
\begin{proof}
To prove the theorem we must show that the identity element of $\Ga$ is contained in $\Ha$, and whenever $g_1$ and $g_2$ are in $\Ha$, then so are $h_1^{-1}$ and $h_1 h_2$, so the elements of $\Ha$indeed form a group. We prove the theorem for $\Ha$. The proof for $\Hb$ is identical.  

{Identity}. Let $e_\alpha$ be the identity of $\Ga$ and $e_\beta$ be the identity of $\Gb$.  We see that $e_alpha$ is in $H^alpha$, since
\[
e_\alpha A^{\alpha\beta} = A^{\alpha\beta} = A^{\alpha\beta} e_\beta. 
\]

\emph{Inverse existence}. For all $g\in \Ha$ exist $g^{-1}\in \Ha$. By definition, if $g\in \Ha$ then there exists $h$ such that $g A^{\alpha\beta}=A^{\alpha\beta} h$. Consider the permutation $g^{-1}$, then
\[
g A^{\alpha\beta} = A^{\alpha\beta} h
\quad
\Rightarrow
\quad
g^{-1} g A^{\alpha\beta} = g^{-1} A^{\alpha\beta} h
\quad
\Rightarrow
\quad
A^{\alpha\beta} h^{-1} h = g^{-1} A^{\alpha\beta} h
\]
thus obtaining $A^{\alpha\beta} = g^{-1} A^{\alpha\beta} h \Rightarrow g^{-1} A^{\alpha\beta}=A^{\alpha\beta} h^{-1}$. Since $\Gb$ is a group, $h\in \Gb \Rightarrow h^{-1}\in\Gb$, thus proving that $g^{-1}\in  \mathcal{H}^\alpha$.

\emph{Closure}. We need to prove that if $g_1$ and $g_2$ are in $\Ha$, then the product $g_1 g_2$ is also in $\Ha$. By definition, since $g_1,\ g_2\in\Ha$, there exist $h_1,\ h_2$ such that $g_1 A^{\alpha\beta} = A^{\alpha\beta} h_1$ and $g_2 A^{\alpha\beta} = A^{\alpha\beta} h_2$. Then
\[
g_1 g_2 A^{\alpha\beta} = g_1 A^{\alpha\beta} h_2 = A^{\alpha\beta} h_1 h_2.
\]
Since $\Gb$ is a group, $h_1,h_2\in\Gb \Rightarrow h_1 h_2 \in \Gb$, thus proving that $g_1g_2\in\Ha$.
\end{proof}

A consequence of the theorem is that $g \in \Ha \Rightarrow g \in \Ga$: thus, to find $\Ha$, we have simply to take all the permutations $h\in\Gb$ and check which of the permutations in $\Ga$ respect the conjugacy relation \eqref{eq:conj}. 
Similarly, we can define the conjugate set Eq.~(\ref{eq:Hb}), find its elements in the same way and show that it is a subgroup of $\Gb$.

Since we are interested here only in undirected networks, we prove a corollary that simplifies the computations of $\Ha$ and $\Hb$.
\begin{corollary}
For the case of undirected networks we only need one of the following relationships to prove that $\Ha$ and $\Hb$ are subgroups,
(i) $g A^{\alpha\beta}=A^{\alpha\beta}h$
or 
(ii) $h A^{\beta\alpha}=A^{\beta\alpha}g$.
\end{corollary}
\begin{proof}
Let's choose to use relationship (i) to define $\Ha$ and $\Hb$. This means eliminating requirement (ii) in the $\Ha$ definition and replacing (ii) with (i) in the  $\Hb$ definition.
We can use the same logic we used in the previous theorem to show that these definitions of $\Ha$ and $\Hb$ also produce subgroups. The question is, are they the same as those of Eqs.~(\ref{eq:Ha}) and (\ref{eq:Hb})?  We can show that this is the case using the fact that for undirected coupling the overall coupling matrix is symmetric.  Hence, off-diagonal interlayer coupling blocks are related as in $A^{\beta\alpha}=(A^{\alpha\beta})^T$.
For permutations $h^{-1}=h^T$, etc. Hence, since $\Ha$ and $\Hb$ are groups, $h^{-1}$ and $g^{-1}$ also obey the defining equation (i), so
\begin{equation} \label{eq:hAba_Abag}
\begin{split}
g^{-1} A^{\alpha\beta}=A^{\alpha\beta}h^{-1}\iff g^TA^{\alpha\beta}=A^{\alpha\beta}h^T \iff (A^{\alpha\beta})^Tg=h (A^{\alpha\beta})^T \\
\iff (A^{\alpha\beta})^Tg=h (A^{\alpha\beta})^T\iff h A^{\beta\alpha}=A^{\beta\alpha}g
\end{split}
\end{equation}
\end{proof}
A consequence of this corollary is that for the case of undirected networks, if $h\in\Gb$ is involved in a conjugacy relation for some $g$, then it is in $\Hb$. As a result, one does not need to search for both $\Ha$ and $\Hb$, but populating $\Ha$ automatically populates $\Hb$. 

\begin{theorem}
\label{numelemsEqClassK}
$\mathcal K^\alpha_1$ is a subgroup of $\mathcal {H}^\alpha$.
\end{theorem}
\begin{proof}

{\em Identity}. The identity $e_\alpha$ is in $\mathcal K^\alpha_1$ by definition.  

{\em Inverse}. Let $ g \in \mathcal K^\alpha_1$ then, $e_\alpha A^{\alpha\beta}=g^{-1}g A^{\alpha\beta}=A^{\alpha\beta}$. But we also have $g A^{\alpha\beta}=A^{\alpha\beta}$, hence, $g^{-1} A^{\alpha\beta}=A^{\alpha\beta}$ so $g^{-1} \in \mathcal K^\alpha_1$.

{\em Closure}. Let $g_1 \text{ and } g_2 \in \mathcal K^\alpha_1$ then, $g_1  g_2  A^{\alpha\beta}= g_1  A^{\alpha\beta}= A^{\alpha\beta} \implies g_1  g_2 \in \mathcal K^\alpha_1$.

Hence, $\mathcal K^\alpha_1$ is a subgroup of $\mathcal {H}^\alpha$.

\end{proof}

The same reasoning holds for the $\beta$ layer, etc.  Next we state an additional property of the subgroups $\mathcal K^\alpha_1$, $\mathcal K^\beta_1$, etc.:

\begin{corollary}
$\mathcal K^\alpha_1$ is a normal subgroup. This means that if $g_1 \in \mathcal K^\alpha_1$, then $\forall g \in \mathcal {H}^\alpha$ we have $ g^{-1} g_1 g \in \mathcal K^\alpha_1$ or, equivalently, $ g^{-1} \mathcal K^\alpha_1 g = \mathcal K^\alpha_1$.
\end{corollary}
\begin{proof}
We know $ \forall g \in \mathcal {H}^\alpha$, if $g A^{\alpha\beta}=A^{\alpha\beta} h$, then $g^{-1} A^{\alpha\beta}=A^{\alpha\beta} h^{-1}$. Hence, $\forall g_1 \in \mathcal K^\alpha_1$, 
$g^{-1} g_1 g A^{\alpha\beta}=g^{-1} g_1 A^{\alpha\beta}h=g^{-1} A^{\alpha\beta}h=
A^{\alpha\beta} h^{-1}h=A^{\alpha\beta} \implies g^{-1} g_1 g \in K^\alpha_1$ and $K^\alpha_1$ is a normal subgroup.
\end{proof}

And, finally, we prove the following corollary. 

\begin{corollary}
All the $\mathcal K^\alpha_i$, $i \neq 1$, are left and right cosets of $\mathcal K^\alpha_1$.
\end{corollary}
\begin{proof}
If $g_1 \in \mathcal K^\alpha_1$ and $g \in \mathcal K^\alpha_i$ with $gA^{\alpha\beta}=A^{\alpha\beta}h$ , then $g g_1 A^{\alpha\beta}= g A^{\alpha\beta}=A^{\alpha\beta}h \implies g g_1 \in \mathcal K^\alpha_i$. Similarly we can prove $g_1 g \in \mathcal K^\alpha_i$. Since group products are unique, i.e. $g g_1 = g g_2 \implies g_1 = g_2$, we have $g \mathcal K^\alpha_1= \mathcal K^\alpha_i= \mathcal K^\alpha_i g$. So , $\mathcal K^\alpha_i$ are left and right cosets of $\mathcal K^\alpha_1$. This also means all the $\mathcal K^\alpha_i$ have the same number of elements.
\end{proof}

\subsection*{Properties of the Equivalence Classes and structure of the matrix $T$}

Using the equivalence classes of each layer we can show the following is true:  the $T$ matrix for the entire multilayer system is of a block diagonal form:

\begin{equation}\label{eq:TblockMatrix}
  T = 
  \begin{pmatrix}
    T^\alpha & 0 \\
    0 & T^\beta \\
  \end{pmatrix}
\end{equation}
Can we find the matrices of each block independently?  This may seem intuitively obvious, but it is not obvious that the arithmetic will work out. Consider the two layer system $\alpha$ and $\beta$ as above. The group $\mathcal G$ of the full system is given by the union of direct products such as $\mathcal K^\alpha_i \times \mathcal K^\beta_i$. If $n_\alpha$ is the number of elements in $\mathcal K^\alpha_i$ and $n_\beta$ is the number of elements in $\mathcal K^\beta_i$, then the number of elements in an equivalence class direct product is $n_\alpha n_\beta$. And if $K$ is the number of equivalence classes, then the number of terms in the whole direct product is $K n_\alpha n_\beta$. 

In the calculation of the $T$ matrix for the entire system we form the projection operators $P^{(l)}$ for each of the IRR labeled by $l$ using the sums \cite{pecora2014cluster}

\begin{equation}
P^{(l)}= \frac{d^{(l)}}{d} \sum_{\cal C} \mu^{(l)}_{\cal C} \sum_{g \in \cal C} g
\end{equation}
where $\cal C$ is a conjugacy class, $\mu^{(l)}_{\cal C}$ is the character of that class for the $l$th IRR, $d^{(l)}$ is the dimension of the $l$th IRR and $d$ is the order (size) of the group. For the multilayer system $d=K n_\alpha n_\beta$. The elements of the $\beta$ level group appear in the sum $n_\alpha$ times and the elements of the $\alpha$ level group appear in the sum $n_\beta$ times. This will contribute factors to the sum for each layer. However, because $d=K n_\alpha n_\beta$ is in the denominator  the extra factors in the numerator will be cancelled in each layer leaving the correct dimension divisor for each layer's sum. Hence, we can find the $T$ matrix for the whole system by finding the $T$ matrices ($T^\alpha$ and $T^\beta$) for each layer independently and putting them into the block form to construct $T$.

\subsection*{Computing the symmetry group of a simple multilayer network}\label{simpleExample}

Let us consider the simple example of a two-layer undirected network shown in Fig.\ \ref{fig:twolayer1}. 
We show in detail the calculations to determine the final group $\mathcal G$ of the full network.

The group of symmetries of the two layers are (written in cyclic notation) 

\begin{eqnarray}
\Ga&=&\{(),(1,2)(3,4),(2,4),(1,4,3,2),(1,2,3,4),(1,3),(1,4)(2,3),(1,3)(2,4)\},\\
\Gb&=&\{(),(1,3)\},
\end{eqnarray}

\noindent where, for example, (1,4)(2,3) means move node 1 to node 4 (and 4 to 1) and 3 to 2 (and 2 to 3). Also, (1,4,3,2) means move 1 to 4, 4 to 3, 3 to 2, and 2 to 1.

From these we determine $\Ha$ and $\Hb$ using the results of Theorem 1. To find $\Ha$, we take all the  permutations $h \in \Gb$ and we check  which  of  the  permutations  in $\Ga$ respect  the  conjugacy  relation \eqref{eq:conj}. 
We obtain 
\begin{eqnarray}
\Ha&=&\{(),(1,2)(3,4),(1,4)(2,3),(1,3)(2,4)\},\\
\Hb&=&\{(),(1,3)\},
\end{eqnarray}

Now we must define the $\sim$ relation. Applying the permutations in $\Ha$ to $A^{\alpha\beta}$ we obtain
\[
()\begin{bmatrix}
    1 & 0 & 0 \\
    1 & 0 & 0 \\
    0 & 0 & 1 \\
    0 & 0 & 1 \\
    0 & 0 & 0
    \end{bmatrix} = 
    (1,2)(3,4) \begin{bmatrix}
    1 & 0 & 0 \\
    1 & 0 & 0 \\
    0 & 0 & 1 \\
    0 & 0 & 1 \\
    0 & 0 & 0
    \end{bmatrix} =
    \begin{bmatrix}
    1 & 0 & 0 \\
    1 & 0 & 0 \\
    0 & 0 & 1 \\
    0 & 0 & 1 \\
    0 & 0 & 0
    \end{bmatrix} \quad \Longrightarrow \quad
    \begin{array}{c}
    () \sim (1,2)(3,4)\\[2mm]
    \mathcal K_1^\alpha = \{(), (1,2)(3,4)\}
    \end{array}
\]
\[
    \qquad
    (1,4)(2,3)\begin{bmatrix}
    1 & 0 & 0 \\
    1 & 0 & 0 \\
    0 & 0 & 1 \\
    0 & 0 & 1 \\
    0 & 0 & 0
    \end{bmatrix}=(1,3)(2,4)\begin{bmatrix}
    1 & 0 & 0 \\
    1 & 0 & 0 \\
    0 & 0 & 1 \\
    0 & 0 & 1 \\
    0 & 0 & 0
    \end{bmatrix} = 
    \begin{bmatrix}
    0 & 0 & 1 \\
    0 & 0 & 1 \\
    1 & 0 & 0 \\
    1 & 0 & 0 \\
    0 & 0 & 0
    \end{bmatrix} \quad\Longrightarrow\quad 
    \begin{array}{c}
   (1,4)(2,3) \sim (1,3)(2,4)\\[2mm]
    \mathcal K_2^\alpha = \{(1,4)(2,3), (1,3)(2,4)\}
    \end{array}
\]
while applying the permutation of $\Hb$ to $A^{\alpha\beta}$ and recalling the conjugacy relation $g A^{\alpha\beta} = A^{\alpha\beta} h$ we obtain
\[
\begin{bmatrix}
    1 & 0 & 0 \\
    1 & 0 & 0 \\
    0 & 0 & 1 \\
    0 & 0 & 1 \\
    0 & 0 & 0
    \end{bmatrix} () = \begin{bmatrix}
    1 & 0 & 0 \\
    1 & 0 & 0 \\
    0 & 0 & 1 \\
    0 & 0 & 1 \\
    0 & 0 & 0
    \end{bmatrix} \quad \Longrightarrow \quad
    \mathcal K_1^\beta = \{()\}
\]
\[
\begin{bmatrix}
    1 & 0 & 0 \\
    1 & 0 & 0 \\
    0 & 0 & 1 \\
    0 & 0 & 1 \\
    0 & 0 & 0
    \end{bmatrix} (1,3) = \begin{bmatrix}
    0 & 0 & 1 \\
    0 & 0 & 1 \\
    1 & 0 & 0 \\
    1 & 0 & 0 \\
    0 & 0 & 0
    \end{bmatrix} \quad \Longrightarrow \quad
    \mathcal K_2^\beta = \{(1,3)\}.
\]

Combining the permutations from each pair of disjoint subsets as in Eq. \eqref{eq:2Dblockdiag}, we obtain the full group,
\begin{equation} \label{eq:Eg1Group}
\mathcal G=
\biggl\{
\begin{pmatrix}
()&0 \\
0&()
\end{pmatrix},
\begin{pmatrix}
(1,2)(3,4)&0 \\
0&()
\end{pmatrix},
\begin{pmatrix}
(1,4)(3,2)&0 \\
0&(1,3)
\end{pmatrix},
\begin{pmatrix}
(1,3)(2,4)&0 \\
0&(1,3)
\end{pmatrix}
\biggr\}.
\end{equation}

Since $\mathcal K_1^\alpha$ and $\mathcal K_1^\beta$ are the subgroups of $\mathcal H^\alpha$ and $\mathcal H^\beta$, respectively, this means that the separate clusters in the $\alpha$ layer, $[1,2]$ and $[3,4]$, are ILS clusters and the following cluster bifurcations cause no bifurcations in the possible $\beta$ dynamics or synchronous clusters (again, stability is a separate issue):

\begin{equation} \label{eq:ghindeplayersym}
 [1,2,3,4] \rightarrow [1,3] \text{ and } [2,4],
\end{equation}
or
\begin{equation} \label{eq:ghindeplayersym}
 [1,2,3,4] \rightarrow [1,4] \text{ and } [2,3],
\end{equation}
which would still allow [1,3] to be synchronous in the $\beta$ layer.
However, the bifurcation 
\begin{equation} \label{eq:ghindeplayersym}
 [1,2][3,4] \rightarrow [1], [2], [3],  \text{ and } [4]
\end{equation}
would break a symmetry in the coset $\mathcal K_2^\alpha$ if the nodes  $1$ and $3$ were synchronized. But if they are not synchronized, then the system is operating in a subgroup of the original group and only the identity would remain for the $\beta$ layer. This means in this case $1$ and $3$ are each in their own (trivial) cluster already in the $\beta$ layer and that does not prevent either state $[1,2][3,4]$ or $[1], [2], [3], [4]$ in the $\alpha$ layer.  These cases can be easily seen from Fig.\ \ref{fig:twolayer1} since node $1$ in the $\beta$ layer only depends on the sum of nodes $1$ and $2$ in the $\alpha$ layer.  And similarly for the dependence of node $3$ in the $\beta$ layer on nodes $3$ and $4$ in the $\alpha$ layer. However, while these ILS relations are easy to see in this simple case, it would require the calculation of $\mathcal K_1^\alpha$ and $\mathcal K_1^\beta$ for more complicated networks to find the ILS's, associated clusters, and their allowed bifurcations.

\subsection*{Symmetries of multiplex and multidimensional networks}

Here we describe how to compute the group of symmetries of multiplex and multidimensional networks, which are special cases of the general multilayer network problem presented in the Results Section. We show that the problems of computing the group of symmetries of a multiplex network and of a multidimensional network are closely related to each other. In what follows, we start by considering the problem of computing the symmetry group of a multiplex network and mapping it to the problem of computing the symmetry group of a multidimensional network.

Multiplex networks \cite{verbrugge1979multiplexity,sola2013eigenvector,gomez2013diffusion} are a particular class of multilayer networks;
in multiplex networks, different features of the same set of agents are described in each different layer (\eg in a social system, each layer represents the opinion of a person on a different topic, and links capture how the different social interactions influence the person thinking on each topic).
A multiplex network is thus formed of several layers, with each layer 
containing the same number of nodes $\mathcal{N}$. 
Moreover, since interlayer coupling only occurs between node $i$ in a given layer and the same node $i$ in a different layer,  $A^{\alpha\beta}=\mathbb{I}_\mathcal{N}$, where $\mathbb{I}_\mathcal{N}$ is the $\mathcal{N}$-dimensional identity matrix. 

A general multiplex network is governed by the following set of equations \cite{tang2017master,Leyva2017},
\begin{equation}
\dot{\bf x}_i^\alpha={\bf F}^\alpha({\bf x}_i^\alpha)+ \sigma^{\alpha\alpha} \sum_{j=1}^\mathcal{N} A_{ij}^{\alpha\alpha} {\bf H}^{\alpha\alpha}({\bf x}_j^\alpha)+ \sum_{\beta \neq \alpha} \sigma^{\alpha\beta} {\bf H}^{\alpha\beta}({\bf x}_i^\beta),  
\label{eq: multiplex eqn}
\end{equation}
$i=1,...,\mathcal{N}$, where the matrix $A^{\alpha\alpha}$ represents the intra-layer connectivity, the function ${\bf F}^{\alpha}$ represents the intrinsic dynamics of each node inside a layer  and the functions ${\bf H}^{\alpha\alpha}$ and ${\bf H}^{\alpha\beta}$ represent the form of the intra- and inter-layer coupling, respectively. 

We reduce the set of Eqs.\ \eqref{eq: multiplex eqn} to a much more compact form and use it to compute the symmetries of a multiplex network.
Introducing
\begin{equation}
{\bf x}_i=\begin{bmatrix}
    {\bf x}_i^\alpha \\ {\bf x}_i^\beta \\ \vdots
\end{bmatrix}, \quad
{\bf F}({\bf x}_i)= 
 \begin{bmatrix}
   {\bf F}^\alpha({\bf x}_i^\alpha) + \sum_{\lambda \neq \alpha} \sigma^{\alpha \lambda} {\bf H}^{\alpha\lambda}({\bf x}_i^\lambda) \\
   {\bf F}^\beta({\bf x}_i^\beta) + \sum_{\lambda \neq \beta} \sigma^{\beta \lambda} {\bf H}^{\beta\lambda}({\bf x}_i^\lambda) \\
   \vdots  
  \end{bmatrix}, \quad
{\bf H}^{\lambda}({\bf x}_i)= 
 \begin{bmatrix}
   \delta_{\lambda\alpha} {\bf H}^{\alpha\alpha}({\bf x}_i^\alpha)\\
   \delta_{\lambda\beta} {\bf H}^{\beta\beta}({\bf x}_i^\beta) \\
   \vdots \\
  \end{bmatrix}, \quad
  \begin{array}{c}
  \sigma^{\lambda}=\sigma^{\lambda\lambda} \\
  A^\lambda = A^{\lambda\lambda}
 \end{array}
\end{equation}
where $\lambda\in\{\alpha,\beta,\ldots\}$, $\delta_{\alpha \lambda}$ is the Kronecker delta (e.g., $\delta_{\lambda\alpha}=1$ if $\lambda=\alpha$, 0 otherwise) we can rewrite Eq.\ \eqref{eq: multiplex eqn} as
\begin{equation} \label{eq:CompactEqns_SL}
    \dot {\bf x} = {\bf F}({\bf x}) +
    \sum_{\lambda=1}^{\Lambda} \sigma^{\lambda}
     A^{\lambda} {\bf H}^{\lambda} ({\bf x}),
\end{equation}
which is the equation of a {multidimensional network} \cite{berlingerio2013multidimensional,coscia2013you}, i.e., a network with only one layer but different types of interactions. Each interaction type is described by a different adjacency matrix $A^\lambda$, $\lambda=1,...,\Lambda$.  Blaha et al. \cite{blaha2019cluster} recently presented a simple experimental realization of such a network. 

We have shown the equivalence between a multiplex network and a multidimensional  network. We now discuss how to compute the symmetries of the multidimensional network \eqref{eq:CompactEqns_SL}, both analytically and computationally.
The analytic approach gives insight into the origins of the final symmetry permutation group. 
The software approach allows a direct calculation of the full group, but without the insights into its structure.

To analytically define the group of symmetries of the multidimensional network, we introduce $\mathcal{G}^{\lambda}$ as the group of symmetries for each interaction type $\lambda$. 
Each element of the group $\mathcal{G}^{\lambda}$ can be represented by a permutation matrix $\Pi$ that commutes with $A^{\lambda}$, $\Pi A^{\lambda}=A^{\lambda} \Pi$.  
Then, the symmetry group of the whole multidimensional network $\mathcal{G}$ is  given by $\mathcal{G}^{1} \cap \mathcal{G}^{2} ... \cap \mathcal{G}^{\Lambda} $, which is the largest common subgroup of the $\Lambda$ symmetry groups $\{\mathcal{G}^{\lambda}\}$. 

Computationally, the group of symmetries of the multidimensional network $\mathcal{G}$ can be found using available computational group theory tools, like GAP or SAGE \cite{gap2008,sage2006}. 
This software computes the group of symmetries of a labeled graph \cite{mckay1981practical,mckay2014practical};
a multidimensional network can easily be remapped to a labeled graph. 
To remap the network, we define the labeled adjacency matrix $A^{\text{lab}}$.
Matrix entries $A_{ij}^{\text{lab}}$ are defined by how pairs of nodes $i$ and $j$ interact.
For each pair of nodes, one of three cases occurs: 
(i) there is no interaction between $i$ and $j$, 
(ii) there is one type of weighted interaction between $i$ and $j$ and 
(iii) there are two or more types of weighted interactions between $i$ and $j$.
If there is no interaction between node $i$ and node $j$ (case i), then $A_{ij}^{\text{lab}}=0$. 
We represent a single edge (case ii) with a $2$-tuple $(\tau,w)$, where the integer $\tau$ is the edge type and the real number $w$ is the edge weight.
We represent a multiple edge formed by $q$ interactions (case iii) with a $2q$-tuple, $(\tau_1,w_1,\tau_2, w_2,,...,\tau_q,w_q)$, $\tau_1<\tau_2<...<\tau_q$, where each pair $(\tau_i,w_i)$ represents an interaction type and the associated weight. 
We can partition the set of the network edges into $Z$ subsets of edges that are all represented by the same tuple. 
Then the entries of the matrix $A^{\text{lab}}$ are such that $A^{\text{lab}}_{ij}=0$ if there is no interaction between nodes $i$ and $j$ and $A^{\text{lab}}_{ij}=z$, $z=1,...,Z$ if the edge between nodes $i$ and $j$ belongs to the subset $z$. 
This forces the software to consider only permutations that involve edges of the same type and of the same weight, i.e., permutation matrices $\Pi$ that commute with  $A^{\text{lab}}$, $\Pi A^{\text{lab}}=A^{\text{lab}} \Pi$. 
As a result, the computed permutations $\Pi$ that form $\mathcal G^\lambda$  also commute with all the adjacency matrices of the multidimensional network $A^{1}, \ldots, A^{\Lambda}$: $\mathcal G$ is thus the largest common subgroup of the $\Lambda$ symmetry groups.


\subsection*{A network with three layers}
\label{sec:threelayers}

In the Results Section we have considered multilayer networks with not more than two layers and only one type of inter-layer coupling. Here we show how our procedure can be generalized to  a multilayer netwotk with $M=3$ layers. Later we will consider the case of any number $M$ of layers.
For the case $M=3$,  the vector field is
\begin{equation}
\begin{array}{rl}
	\dot{\bx}^\alpha &= {\bf F}^\alpha(\bx^\alpha) + \sigma^{\alpha\alpha} A^{\alpha \alpha} \bH^{\alpha\alpha} (\bx^\alpha) 
	+ \sigma^{\alpha \beta} A^{\alpha \beta} \bH^{\alpha \beta} (\bx^\beta)
	+ \sigma^{\alpha \gamma} A^{\alpha \gamma} \bH^{\alpha \gamma} (\bx^\gamma)
	\\
	\dot\bx^\beta  &= {\bf F}^\beta(\bx^\beta) + \sigma^{\beta\beta} A^{\beta \beta} \bH^{\beta\beta} (\bx^\beta) 
	+ \sigma^{\beta \alpha} A^{\beta \alpha} \bH^{\beta \alpha} (\bx^\alpha)
	+ \sigma^{\beta \gamma} A^{\beta \gamma} \bH^{\beta \gamma} (\bx^\gamma)
	\\
	\dot\bx^\gamma  &= {\bf F}^\gamma(\bx^\gamma) + \sigma^{\gamma\gamma} A^{\gamma \gamma} \bH^{\gamma\gamma} (\bx^\gamma) 
	+ \sigma^{\gamma \alpha} A^{\gamma \alpha} \bH^{\gamma \alpha} (\bx^\alpha)
	+ \sigma^{\gamma \beta} A^{\gamma \beta} \bH^{\gamma \beta} (\bx^\beta)	
\end{array}
\end{equation}
If we apply three permutations $g \in \Ga$, $h \in \Gb$, and $k \in \mathcal{G}^\gamma$ on the system we obtain
\begin{equation}
\begin{array}{rl}
	g\dot{\bx}^\alpha &= {\bf F} ^\alpha(g\bx^\alpha) + \sigma^{\alpha\alpha} A^{\alpha \alpha} \bH^{\alpha\alpha} (g\bx^\alpha) 
	+ \sigma^{\alpha \beta} gA^{\alpha \beta} \bH^{\alpha \beta} (\bx^\beta)
	+ \sigma^{\alpha \gamma} gA^{\alpha \gamma} \bH^{\alpha \gamma} (\bx^\gamma)
	\\
	h\dot\bx^\beta  &= {\bf F}^\beta(h\bx^\beta) + \sigma^{\beta\beta} A^{\beta \beta} \bH^{\beta\beta} (h\bx^\beta) 
	+ \sigma^{\beta \alpha} hA^{\beta \alpha} \bH^{\beta \alpha} (\bx^\alpha)
	+ \sigma^{\beta \gamma} hA^{\beta \gamma} \bH^{\beta \gamma} (\bx^\gamma)
	\\
	k\dot\bx^\gamma  &= {\bf F}^\gamma(k\bx^\gamma) + \sigma^{\gamma\gamma} A^{\gamma \gamma} \bH^{\gamma\gamma} (k\bx^\gamma) 
	+ \sigma^{\gamma \alpha} k A^{\gamma \alpha} \bH^{\gamma \alpha} (\bx^\alpha)
	+ \sigma^{\gamma \beta} k A^{\gamma \beta} \bH^{\gamma \beta} (\bx^\beta).
\end{array}
\end{equation}

For a directed graph, we have 6 conjugacy relationships to be satisfied:
\[
gA^{\alpha \beta} = A^{\alpha \beta} h,\ gA^{\alpha \gamma} = A^{\alpha \gamma}k,\ hA^{\beta \alpha}=A^{\beta \alpha}g,\ hA^{\beta \gamma}=A^{\beta \gamma}k,\  k A^{\gamma \alpha} = A^{\gamma \alpha}g,\text{ and } k A^{\gamma \beta} = A^{\gamma \beta}h.
\]
If the graph is undirected, then half are redundant (\eg $gA^{\alpha \beta} = A^{\alpha \beta} h$ is the same relationship as $hA^{\beta \alpha}=A^{\beta \alpha}g$).
From these relationships we can define
\begin{equation}\label{eq:HaHbHc3lyrs}
\begin{array}{rcl}
\mathcal{H}^\alpha &=
\left\{ g \in \mathcal{G}^\alpha | \exists h \in \mathcal{G}^\beta\mbox{ and }k \in \mathcal{G}^\gamma :g A^{\alpha \beta} = A^{\alpha \beta }h \mbox{ and }   g A^{\alpha \gamma} = A^{\alpha \gamma}k \right\}, \\
\mathcal{H}^\beta &= \left\{ h \in \mathcal{G}^\beta | \exists g \in \mathcal{G}^\alpha \mbox{ and  } k \in \mathcal{G}^\gamma:hA^{\beta \alpha}=A^{\beta \alpha}g \mbox{ and } hA^{\beta \gamma}=A^{\beta \gamma}k \right\},\\
\mathcal{H}^\gamma &= \left\{ k \in \mathcal{G}^\gamma | \exists g \in \mathcal{G}^\alpha \mbox{ and }h \in \mathcal{G}^\beta:
 k A^{\gamma \alpha} = A^{\gamma \alpha}g
\mbox{ and } k A^{\gamma \beta} = A^{\gamma \beta}h \right\}.
\end{array}
\end{equation}
Using the same reasoning as in Theorem \ref{subgroup}, each subset in Eqs.~\eqref{eq:HaHbHc3lyrs} is a subgroup of its layer's group. 
The equivalence relationship $\sim$ is
now  defined by a set of two equations that must be satisfied, i.e.,
\[
g_1 \sim g_2 \quad\Longleftrightarrow\quad g_1 A^{\alpha\beta} =g_2 A^{\alpha\beta} \mbox{ and } g_1 A^{\alpha\gamma} =g_2 A^{\alpha\gamma} 
\]
thus defining
\begin{equation}
    \begin{array}{lcr}
\mathcal{K}^\alpha_{ij} &=& \{g \in \Ha: g A^{\alpha\beta}=M_i,\ g A^{\alpha\gamma}=M_j\},\\
\mathcal{K}^\beta_{ij} &=& \{h \in \Hb: h A^{\beta\alpha}=M_i,\ h A^{\beta\gamma}=M_j\},\\
\mathcal{K}^\gamma_{ij} &=& \{k \in \mathcal{H}^\gamma: k A^{\gamma\alpha}=M_i,\ k A^{\gamma\beta}=M_j\}.
\end{array}
\end{equation}

There are $K_1K_2$ disjoint sets for each subgroup $\mathcal{H}$, where
$K_1$ is the number of different $M_i$ and
$K_2$ is the number of different $M_j$ that can be obtained applying any of the compatible permutations.  
Finally, we form the final symmetry group of the multilayer network using the disjoint sets as in Eq.~\eqref{eq:2DGroup}, namely,
\begin{equation} \label{eq:3DGroup}
\mathcal G=\left \{
\begin{pmatrix}
g&0&0 \\
0&h&0 \\
0&0&k
\end{pmatrix}
\Big |  g \in \mathcal K^{\alpha}_{ij} \mbox{ , } h \in \mathcal K^{\beta}_{ij},  \mbox{ and } k \in \mathcal K^{\gamma}_{ij}\mbox{ for } i=1,...,K_1, j=1,...,K_2
\right \}.
\end{equation}

{As before with two layers the ILSs for this three-layer case will be in each of the layers' equivalence class subgroups, $\mathcal{K}^\alpha_{ij}$, $\mathcal{K}^\beta_{ij}$, and $\mathcal{K}^\gamma_{ij}$.}

\subsection*{Two interlayer coupling types}
\label{sec:twocouplings}
For the case of two or more coupling types between two layers, the extension is similar to the previous section. 
The general equation in this case is
\begin{equation}
    \begin{array}{lcr}
	\dot{\bx}^\alpha &=& {\bf F}^\alpha(\bx^\alpha) + \sigma^{\alpha\alpha} A^{\alpha \alpha} \bH^{\alpha\alpha} (\bx^\alpha) 
	+ \sigma^{\alpha \beta,1} A^{\alpha \beta,1} \bH^{\alpha \beta,1} (\bx^\beta)
	+ \sigma^{\alpha \beta,2} A^{\alpha \beta,2} \bH^{\alpha \beta,2} (\bx^\beta)
	\\
	\dot\bx^\beta  &=& {\bf F}^\beta(\bx^\beta) + \sigma^{\beta\beta} A^{\beta \beta} \bH^{\beta\beta} (\bx^\beta) 
	+ \sigma^{\beta \alpha,1} A^{\beta \alpha,1} \bH^{\beta \alpha,1} (\bx^\alpha)
	+ \sigma^{\beta \alpha,2} A^{\beta \alpha,2} \bH^{\beta \alpha,2} (\bx^\alpha).
\end{array}
\end{equation}

Applying two permutations $g\in\Ga$ and $h\in\Gb$ to the system, we obtain 4 conjugacy relationships to be satisfied:
\[
g A^{\alpha \beta,1} = A^{\alpha \beta,1} h,\quad
g A^{\alpha \beta,2} = A^{\alpha \beta,2} h,\quad
h A^{\beta \alpha,1} = A^{\beta \alpha,1} g,\quad
h A^{\beta \alpha,2} = A^{\beta \alpha,2} g,
\], 
 two of which are redundant for undirected graphs. 
The $\mathcal{H}$ groups are still defined by two of these relationships
\begin{equation}
\begin{array}{rcl}
\mathcal{H}^\alpha &=&
\left\{ g \in \mathcal{G}^\alpha | \exists h \in \mathcal{G}^\beta:g A^{\alpha \beta,1} = A^{\alpha \beta,1}h \mbox{ and }   g A^{\alpha \beta,2} = A^{\alpha \beta,2}h \right\}, \\
\mathcal{H}^\beta &=& \left\{ h \in \mathcal{G}^\beta | \exists g \in \mathcal{G}^\alpha: hA^{\beta \alpha,1}=A^{\beta \alpha,1}g \mbox{ and } hA^{\beta \alpha,2}=A^{\beta \alpha,2}g \right\},
\end{array}
\end{equation}
and $g_1\sim g_2$ only when both produce the same left-hand-side for all conjugacy relationships.  
The disjoint $\mathcal{K}$ sets for each subgroup are generated as before (with two indices, since there are two constraints that define the $\sim$ relationship) and the final symmetry group of the multilayer network using the disjoint sets is as in Eq.~\eqref{eq:2DGroup}. The ILSs are
found as above.

\subsection*{Any number of layers and  inter-layer couplings}

Here we present the general case of any number of layers and any number of inter-layer couplings.
Extrapolating from our discussion above, we can define the group of symmetries of a multilayer network with any number of layers and types of coupling.
The number of conjugacy relationships to satisfy grows as the number of unique types of inter-layer coupling grows; 
stated plainly, if we have $\mathcal{A}$ inter-layer adjacency matrices, we have $\mathcal{A}$ conjugacy relationships.
In the case of three layers discussed earlier, there are six such inter-layer couplings; in the case of two layers and two types of interlayer couplings, there are four.
The conjugacy relationships must then be divided in the various layers to define the $\mathcal{H}$ groups, and two permutations in each of the $\mathcal{H}$ groups have the equivalence relationship $\sim$ when they both give the same result in all the left hand side of the $\mathcal{H}$ defining conjugacy relations.
Consequently, the computational effort grows with the number of inter-layer couplings as $O(\mathcal{A})$.


\subsection*{Computation of the symmetry group}


We have already discussed the computation of the symmetry group for the case of multiplex networks and multidimensional networks.
Here we briefly review the computation of the symmetry group for the case of general multilayer networks. 
In the presence of both nodes of different types and edges of different types, a symmetry is only allowed: 
(i) when the edges are interchangeable as discussed before for the case of multidimensional networks and 
(ii) when the nodes that are exchanged are of the same type. 
Available computational group theory software \cite{gap2008,sage2006} allows us to handle both aspects, since they also manage multi-partite labeled graph \cite{mckay1981practical,mckay2014practical}.

To satisfy requirement (i), we provide the software with a labeled supra-adjacency matrix $A^{\text{lab}}$ for the multilayer network, which we define in a similar way as for the previously discussed case of the multidimensional network. 
First, we sequentially renumber all the nodes starting from layer $\alpha$ (so nodes $1,\ldots,N^\alpha$ are the nodes in layer $\alpha$, nodes $N^\alpha+1,\ldots,N^\alpha+N^\beta$ are the nodes in layer $\beta$, and so on).
Then we write the matrix  $A^{\text{lab}}$ that describes the connections between the nodes of the multilayer network: 
$A^{\text{lab}}_{ij}=0$ if there is no (intralayer or interlayer) connections between $i$ and $j$.  
If they are connected, $A^{\text{lab}}_{ij}$  is a tuple as described earlier.

To satisfy requirement (ii), we restrict the search of symmetries to permutations that only involve nodes of the same type (from the same layer).
We thus provide the partition $\mathcal{P}=\big\{ \{1,\ldots,N^\alpha\},\{N^\alpha+1,\ldots,N^\alpha+N^\beta\},\ldots\big\}$ that describes how the renumbered nodes are split between the respective layers. 
This forces the software to consider only permutations that involve nodes of the same type (within the same subset of the partition.)

Finding the group of symmetries of the multilayer network thus reduces to finding  the group of symmetries of the network described by the labeled supra-adjacency matrix $A^{\text{lab}}$, $\Pi A^{\text{lab}}=A^{\text{lab}} \Pi$, with a predefined partition $\mathcal{P}$ of the nodes that are allowed to swap with one another. 

\subsection*{Mathematical model of the experimental system}

Applying Kirchhoff's laws on the circuit, we can write down the following set of differential equations that govern the system's dynamics (we color the coupling terms as in Figure \ref{fig:circuit_mul}): 
\begin{eqnarray*}
\text{Jumper equations:}\nonumber\\
\left(L_{J1}+L_{J2}\right)\dot{I}_J&=&v_{J}^0-V_{J}-I_J R_J\textcolor{orange}{-k_{FJ} L_{F1,J1} \dot{I}_{F,1} -k_{FJ} L_{F2,J2} \dot{I}_{F,2}}\\
C_J \dot{V}_{J}&=&I_J\\
\text{FHN 1:}\nonumber\\
L_{F,1} \dot{I}_{F,1}&=&V_{F,1}-R_{3,1} I_{F,1}\textcolor{orange}{-k_{FJ}L_{F1,J1}\dot{I}_J} \\
C_{F,1} \dot{V}_{F,1}&=&-f\left[V_{F,1}\right]-I_{F,1} \textcolor{blue}{+\frac{1}{R_C}\left[\left(V_{c e,1}-V_{b e,1}-V_{F,1}\right)+\left(V_{c e,2}-V_{b e,2}-V_{F,1}\right)\right]}\\
\text{FHN 2:}\nonumber\\
L_{F,2} \dot{I}_{F,2}&=&V_{F,2}-R_{3,2} I_{F,2}\textcolor{orange}{-k_{FJ}L_{F2,J2}\dot{I}_J} \\
C_{F,2} \dot{V}_{F,2}&=&-f\left[V_{F,2}\right]-I_{F,2} \textcolor{blue}{+\frac{1}{R_C}\left[\left(V_{c e,3}-V_{b e,3}-V_{F,2}\right)+\left(V_{c e,4}-V_{b e,4}-V_{F,2}\right)\right]}\\
\text{Colpitts 1:}\nonumber\\
C_{ce,1} \dot{V}_{c e,1}&=&I_{L,1}-I_c\left[V_{be,1}\right]\textcolor{blue}{+\frac{1}{R_C}\left(V_{F,1}-\left(V_{ce,1}-V_{be,1}\right)\right)}\\
C_{be,1}\dot{V}_{be,1}&=&-\frac{\left(V_{ee}+V_{be,1}\right)}{R_{ee,1}}-I_b\left[V_{be,1}\right]-I_{L,1}\textcolor{blue}{-\frac{1}{R_C}\left[V_{F,1}-\left(V_{ce,1}-V_{be,1}\right)\right]}\\
L_C\dot{I}_{L,1}&=&V_{cc}-V_{ce,1}+V_{be,1}-R_{L,1}I_{L,1}\textcolor{red}{-k_C L_C\dot{I}_{L,2}}\\
\text{Colpitts 2:}\nonumber\\
C_{ce,2} \dot{V}_{c e,2}&=&I_{L,2}-I_c\left[V_{be,2}\right]\textcolor{blue}{+\frac{1}{R_C}\left(V_{F,1}-\left(V_{ce,2}-V_{be,2}\right)\right)}\\
C_{be,2}\dot{V}_{be,2}&=&-\frac{\left(V_{ee}+V_{be,2}\right)}{R_{ee,2}}-I_b\left[V_{be,2}\right]-I_{L,2}\textcolor{blue}{-\frac{1}{R_C}\left[V_{F,1}-\left(V_{ce,2}-V_{be,2}\right)\right]}\\
L_C\dot{I}_{L,2}&=&V_{cc}-V_{ce,2}+V_{be,2}-R_{L,2}I_{L,2}\textcolor{red}{-k_C L_C\dot{I}_{L,1}}\\
\text{Colpitts 3:}\nonumber\\
C_{ce,3} \dot{V}_{c e,3}&=&I_{L,3}-I_c\left[V_{be,3}\right]\textcolor{blue}{+\frac{1}{R_C}\left(V_{F,2}-\left(V_{ce,3}-V_{be,3}\right)\right)}\\
C_{be,3}\dot{V}_{be,3}&=&-\frac{\left(V_{ee}+V_{be,3}\right)}{R_{ee,3}}-I_b\left[V_{be,3}\right]-I_{L,3}\textcolor{blue}{-\frac{1}{R_C}\left[V_{F,2}-\left(V_{ce,3}-V_{be,3}\right)\right]}\\
L_C\dot{I}_{L,3}&=&V_{cc}-V_{ce,3}+V_{be,3}-R_{L,3}I_{L,3}\textcolor{red}{-k_C L_C\dot{I}_{L,4}}\\
\text{Colpitts 4:}\nonumber\\
C_{ce,4} \dot{V}_{c e,4}&=&I_{L,4}-I_c\left[V_{be,4}\right]\textcolor{blue}{+\frac{1}{R_C}\left(V_{F,2}-\left(V_{ce,4}-V_{be,4}\right)\right)}\\
C_{be,4}\dot{V}_{be,4}&=&-\frac{\left(V_{ee}+V_{be,4}\right)}{R_{ee,4}}-I_b\left[V_{be,4}\right]-I_{L,4}\textcolor{blue}{-\frac{1}{R_C}\left[V_{F,2}-\left(V_{ce,4}-V_{be,4}\right)\right]}\\
L_C\dot{I}_{L,4}&=&V_{cc}-V_{ce,4}+V_{be,4}-R_{L,4}I_{L,4}\textcolor{red}{-k_C L_C\dot{I}_{L,3}}
\end{eqnarray*}
The Colpitts nonlinearity arises from its transistor (2N2222) with gain:
\begin{equation} I_b =
  \begin{cases}
    0, &  V_{be} \le V_{th}\\
    \frac{V_{be}-V_{th}}{R_{on}}, &  V_{be} > V_{th}
  \end{cases}
, \qquad I_c=\beta_f I_b.
\end{equation}
The FHN nonlinearity arises from the operational amplifier (AD844) which is arranged to provide a piecewise linear cubic function~\cite{Keener_1983jy} in the form: $f(V_{F,i})=[V_{F,i}-h(V_{F,i})]/R_{1,i}$, where
\begin{equation}\label{eq: fhn_piecewise} h(V_F) =
  \begin{cases}
    -V_{R_+}, &  V_F \le \frac{-R_{4,i}V_{R_+}}{R_{2,i}+R_{4,i}}\\
    \left(\frac{R_{2,i}}{R_{4,i}}+1\right)V_F, &  \frac{-R_{4,i}V_{R_+}}{R_{2,i}+R_{4,i}} < V_F <\frac{R_{4,i}V_{R_+}}{R_{2,i}+R_{4,i}}\\
    V_{R_+}, &  V_F \ge \frac{R_{4,i}V_{R_+}}{R_{2,i}+R_{4,i}}.
  \end{cases}
\end{equation}
Here, $V_{R_+}=6.5V$ is inferred from the FHN amplitude.
Note that the expression $h(V_F)$ is independent of $R_{2,i}$ and $R_{4,i}$ if $R_{2,i}=R_{4,i}$.

 {
 To study the cluster synchronization dynamics of this multilayer system, we vary one coupling and hold the others constant. 
 We vary the Colpitts intralayer coupling, $k_C$, by finely varying the separation of the inductors of the Colpitts. We immobilize Colpitts layer nodes $C_1$ and $C_3$, and move $C_2$ and $C_4$ with a caliper; this precisely and simultaneously changes the $C_1$ to $C_2$ and $C_3$ to $C_4$ separation. 
 We impose the constant Colpitts-FHN intralayer coupling, $R_C$,with a fixed resistor.
We hold the FHN-jumper intralayer coupling, $k_{FJ}$, constant as well; we impose a constant separation between the inductors of the FHN and the jumper with plastic clips. }

{
This experimental testbed differs fundamentally from the numerical examples we consider. 
Although we model our circuit as noise-free, electrical circuits are subject to several kinds of noise, including shot noise, burst noise, and thermal noise~\cite{instruments2007noise,bryantask}. 
Components such as transistors and operational amplifiers have noise profiles, with stronger noise at higher frequencies.
Additionally, we neglect stray resistance, inductance, and capacitance which inevitably arises even in circuits which are carefully crafted to minimize such issues. 
Finally, small parameter mismatches lead the oscillators in each layer to be slightly heterogeneous; uncoupled, the nodes have slightly different frequencies and amplitudes. }


{To write the system in the general form \eqref{eq:gen_dyn},} we first assume identical oscillators on each layer, i.e.,
$L_{J1}=L_{J2}=L_{FJ};$
$~L_{F1,J1}=L_{F2,J2}=L_{FJ};$
$~L_{F,1} = L_{F,2}=L_F;$
$~C_{F,1} = C_{F,2}=C_F;$
$~R_{3,1} = R_{3,2}=R_3;$
$~C_{ce,1}=C_{ce,2}=C_{ce,3}=C_{ce,4}=C_{ce};$
$~C_{be,1}=C_{be,2}=C_{be,3}=C_{be,4}=C_{be};$
$~L_{C,1}=L_{C,2}=L_{C,3}=L_{C,4}=L_{C};$
$~R_{ee,1}=R_{ee,2}=R_{ee,3}=R_{ee,4}=R_{ee}.$
We can then write
\begin{equation}
    \dot {\textbf x}_i^\alpha = F^\alpha({\textbf x}_i^\alpha) +
    \underbrace{ \sigma^{\alpha\alpha}
    \sum_{j=1}^{N^\alpha} A^{\alpha\alpha}_{ij} H^{\alpha\alpha} ({\textbf x}_j^\alpha)}_{\text{intra-layer coupling}}
    +
    \underbrace{\sum_{\beta\neq\alpha}
    \sigma^{\alpha\beta}
    \sum_{j=1}^{N^\beta} A^{\alpha\beta}_{ij} H^{\alpha\beta} ({\textbf x}_j^\beta).}_{\text{inter-layer coupling}}
\end{equation}
Where $x_1^\alpha=I_J$, $x_2^\alpha=V_J$ are the current and the voltage in the jumper;
$x_{i,1}^\beta=I_{F,i}$, $x_{i,2}^\beta=V_{F,i}$ are the current and the voltage of FHN $i$;
$x_{i,1}^\gamma=V_{ce,i}$, $x_{i,2}^\gamma=V_{be,i}$, $x_{i,3}^\gamma=I_{L,i}$ are the voltages and the current of Colpitts $i$.
\\
\\
\textbf{Jumper layer:}
\[
{\textbf F}^\alpha({\textbf x}^\alpha) = 
\begin{pmatrix} 
(R_J x_1^\alpha+x_2^\alpha-v_J^0)/\left(2 L_{FJ} (k_{FJ}^2 L_{FJ}-L_F) \right)\\
x_1^\alpha/C_J
\end{pmatrix},
\]
\[
{\textbf H}^{\alpha\beta}({\textbf x}_i^\beta) = 
\begin{pmatrix} 
2(R_3 x_{i,1}^\beta-x_{i,2}^\beta)\\
0
\end{pmatrix},
\quad
\sigma^{\alpha\beta} = \frac{k_{FJ}}{2(L_F-L_{FJ} k_{FJ}^2)},
\quad
A^{\alpha\beta}=\begin{bmatrix}1 & 1\end{bmatrix}
\]
\textbf{FHN layer}
\[
{\textbf F}^\beta({\textbf x}^\beta) = 
\begin{pmatrix} 
 (-2L_F(R_3 x^\beta_1-x^\beta_2)-k_{FJ}L_F v^0_J)/\left(2 L_F (L_F-k_{FJ}^2 L_{FJ})\right) \\
(-f(x_2^\beta)-x_1^\beta-2 x_2^\beta/R_C)/C_F
\end{pmatrix}, 
\]
\[
{\textbf H}^{\beta\beta}({\textbf x}_j^\beta) = 
\begin{pmatrix} 
R_3 x_{j,1}^\beta-x_{j,2}^\beta\\
0
\end{pmatrix},
\quad
\sigma^{\beta\beta} =\frac{L_{FJ}k_{FJ}^2}{2 L_{F} (L_{FJ}k_{FJ}^2-L_F)},
\quad
A^{\beta\beta}=\begin{bmatrix}
-1 & 1 \\ 1 & -1
\end{bmatrix}
\]
\[
{\textbf H}^{\beta\alpha}({\textbf x}^\alpha) = 
\begin{pmatrix} 
R_J x_{1}^\alpha+x_{2}^\alpha\\
0
\end{pmatrix},
\quad
\sigma^{\beta\alpha} = \frac{k_{FJ}}{2(L_F-L_{FJ}k_{FJ}^2)}=\sigma^{\alpha\beta},
\quad
A^{\beta\alpha}=\begin{bmatrix}1 \\ 1\end{bmatrix}=(A^{\alpha\beta})^T.
\]
\[
{\textbf H}^{\beta\gamma}({\textbf x}^\beta,{\textbf x}^\gamma) = 
\begin{pmatrix} 
0\\
(x_1^\gamma-x_2^\gamma)/C_F
\end{pmatrix},
\quad
\sigma^{\beta\gamma} = \frac{1}{R_C},
\quad
A^{\beta\gamma}=\begin{bmatrix}
1 & 1 & 0 & 0 \\ 
0 & 0 & 1 & 1\end{bmatrix}.
\]
Note that the inductive coupling through the jumper generates a diffusive coupling in the $\beta$ layer (as can be seen in $A^{\beta\beta}$). 
{This is a consequence of the fact that the inductive coupling acts on the derivative of the current $\dot I_{F,i}$, see the orange terms in the Jumper equations. As a result, this induces a direct diffusive coupling between the currents $I_{F,i}$ of the FHN circuits.}
This is represented in figure \ref{fig:circuit_mul} by a dashed line in the $\beta$ layer.

\textbf{Colpitts layer}
\[
{\textbf F}^\gamma({\textbf x}^\gamma) = 
\begin{pmatrix} 
 (x^\gamma_3-I_C(x^\gamma_2)-(x^\gamma_1-x^\gamma_2)/R_C)/C_{ce} \\
(-(V_{ee}+x_2^\gamma)/R_{ee}-I_b(x_2^\gamma)-x_3^\gamma+(x_1^\gamma-x_2^\gamma)/R_C))/C_{be}\\
(x_2^\gamma-x_1^\gamma-R_L x_3^\gamma+V_{CC}(1-k_C))/(L(1-k_C^2))
\end{pmatrix}, 
\]
\[
{\textbf H}^{\gamma\gamma}({\textbf x}_j^\gamma) = 
\begin{pmatrix} 
0\\
0\\
-(x_{j,2}^\gamma-x_{j,1}^\gamma-R_L x_{j,3}^\gamma)
\end{pmatrix},
\quad
\sigma^{\gamma\gamma} = \frac{k_C}{L(1-k_C^2)},
\quad
A^{\gamma\gamma}=\begin{bmatrix}
0 & 1 & 0 & 0 \\
1 & 0 & 0 & 0 \\
0 & 0 & 0 & 1 \\
0 & 0 & 1 & 0 
\end{bmatrix}
\]
\[
{\textbf H}^{\gamma\beta}({\textbf x}^\beta,{\textbf x}^\gamma) = 
\begin{pmatrix} 
x_2^\beta/C_{ce}\\
x_2^\beta/C_{be}\\
0
\end{pmatrix},
\quad
\sigma^{\beta\gamma} = \frac{1}{R_C}=\sigma^{\gamma\beta},
\quad
A^{\gamma\beta}=(A^{\beta\gamma})^T.
\]

\null
\null
\null

\section*{Acknowledgements}
This work was supported by the National Science Foundation through grant (Grant No. 1727948) and the Office of Naval Research through the Naval Research Laboratory’s Basic Research Program. We thank B.Hunt for noting that the equivalence classes of each layer 
form cosets of a subgroup of that layer's group  and M. Hossein-Zadeh and K. Huang for their support with the experimental activities.

\section*{Author contributions} 

F.D.R., L.P., and F.S. worked on the theory. K.B. performed the experiment with the help of F.D.R. A.S. worked on the analysis of real data. I.K. designed  the algorithm for generating multilayer networks with symmetries.  
F.S. supervised all aspects of the project. F.S. wrote the paper with the help of K.B. and L.P.

\section*{Competing interests} 

The authors declare no competing interests.


\end{document}